\documentclass[reqno, 12pt]{amsart}

\usepackage{todonotes}
\usepackage{amssymb}
\usepackage{mathrsfs}
\usepackage{dsfont}
\usepackage{amsmath,amsthm}
\usepackage{amsfonts}
\usepackage{enumerate, xspace}
\usepackage{color}
\usepackage{changebar}
\usepackage{pdfsync}
\usepackage{graphicx}
\usepackage{a4wide}

\numberwithin{equation}{section}

\newtheorem{thm}{Theorem}[section]
\newtheorem{lem}[thm]{Lemma}
\newtheorem{cor}[thm]{Corollary}
\newtheorem{prop}[thm]{Proposition}
\newtheorem{defin}[thm]{Definition}

\newtheorem{hyp}[thm]{Hypothesis}

\newtheorem{rem}[thm]{Remark}

\def\la{{\lambda}}
\def\La{{\Lambda}}
\def\ga{{\gamma}}

\def\om{{\omega}}
\def\Om{{\Omega}}
\def\eps{{\epsilon}}

\newcommand{\bx}{{x}}

\newcommand\nn{\nonumber}

\newcommand{\C}{\mathcal C}
\renewcommand{\C}{{\mathcal X}_n}
\newcommand{\beq}{\begin{eqnarray}}
\newcommand{\eeq}{\end{eqnarray}}
\newcommand{\N}{\mathbb N}
\newcommand{\Z}{\mathbb Z}
\newcommand{\R}{\mathbb R}
\renewcommand{\(}{\left(}
\renewcommand{\)}{\right)}
\renewcommand{\[}{\left[}
\renewcommand{\]}{\right]}
\newcommand{\be}{\begin{equation}}
\newcommand{\ee}{\end{equation}}

\def\barr{{\sigma}}

\begin{document}

\title{Stationary states in infinite volume with non zero current}
\author{Gioia Carinci}
\address{Gioia Carinci, Technische Universiteit Delft, van Mourik Broekmanweg 6, 2628 xe, Delft, The Netherlands}
\author{Cristian Giardin\`a}
\address{Cristian Giardin\`a, University of Modena and Reggio Emilia, via G. Campi 213/b, 41125 Modena, Italy}
\author{Errico Presutti}
\address{Errico Presutti, Gran Sasso Science Institute, viale F. Crispi 7, 67100 L�Aquila, Italy}

\maketitle

\begin{abstract}
\noindent
We study the Ginzburg-Landau stochastic models in infinite domains with some special geometry and  prove that without the help of external forces there are stationary measures with non zero current in three or more dimensions.
\end{abstract}

%

\vspace{1.cm}
\section{Introduction}
Equilibrium statistical mechanics  is based on the paradigm
of the Boltzmann-Gibbs distribution.
This extremely powerful paradigm describes
equilibrium thermodynamics and applies to
a large class of systems, including phase transitions.
By contrast, it does not exist a general and system-independent
approach to non-equilibrium statistical mechanics, where instead
dynamics plays a key role.
The most natural way to create
a non-equilibrium state is by putting an extended system
in contact with two heat or mass reservoirs at different
temperatures or chemical potentials. One could think
of a $d$-dimensional box  $[-N,N]^d$ which identifies the volume of the system and
the two reservoirs are attached to the opposite faces along, say,
the $x$-direction (for simplicity periodic boundary conditions
are chosen in the other  directions).
Due to the reservoirs, the state has a non-zero current in the
$x$-direction. This defines the setting of  boundary-driven systems and the stationary
measure of those systems is then called a {\em non-equilibrium
steady state}. Usually one requires that such a state satisfies
the macroscopic laws of transport, such as the Fourier's law,
by which the heat current is proportional to the gradient of
the temperature, or the Fick's law, implying proportionality between
the mass current and the gradient of the mass density.
As a consequence, the current in a  large system scales
as the inverse of the system length $N$. In particular,
an infinite system ($N\to\infty$) has zero current.

The main question in this paper is about the opposite, namely the existence of stationary
states of infinite systems having a non zero current.  This seems paradoxical
because intuition says that some external forces are needed to sustain the current
which otherwise would die out.  However the fact that the system
has ``a special geometry'' does the trick, as we shall see.
Ruelle  \cite{Ruelle} was the first to give an example of all that by considering a
quantum model describing two infinite systems which interact with each other
via a third finite system. He proved that indeed, in this setting,
there are stationary states with
non zero current.

Later on Gallavotti and Presutti \cite{GP1,GP2,GP3}  studied a similar geometry, namely
a finite system in interaction  with several distinct infinite systems.  The dynamics in \cite{GP1,GP2,GP3} is given by the classical Newton equations with Gaussian thermostatic forces added.
The focus was however on the existence of dynamics in the infinite-volume and the equivalence between Gaussian thermostats and infinite reservoirs.

We will consider here the analogue of the Ruelle model in stochastic systems,
the so called Ginzburg-Landau models.  These are lattice systems with
unbounded (real valued) spins $\phi_x$ called ``charges''.  The dynamics is stochastic
but  it conserves the total charge. It is therefore
a continuous version of the well-known Kawasaki dynamics in the Ising model.
As mentioned, the spatial geometry has an essential role.
The crux of our argument is that in the geometrical set up
that we consider there may exist {\em non-constant} bounded harmonic functions.
We will prove that in such a case there are indeed, in $d\ge 3$ dimensions, infinite-volume
stationary states with non zero current.

For technical reasons we will prove the statement for super-stable Hamiltonians
with non negative, finite range interactions, the class is quite general to include cases where phase transitions are present.  We use such assumptions to prove the existence of the infinite volume dynamics, we believe that they could be relaxed but this is not in the spirit of our paper.

In the case of general Hamiltonians we miss the existence of the infinite-volume dynamics but we can prove that the Fick's law is violated, namely putting the system in contact with two reservoirs which fix the chemical potentials at the right and left faces (as described in the beginning of the introduction) we observe a current which does not decay when the size of the system diverges. See however the remarks after Theorem \ref{thmn3.2}.

In the next section we describe the model, in Section \ref{mainresults} we state the main results which are then proved in the successive sections.

\vskip.5cm

\section{The model}

\subsection{The geometrical setup}
We  consider an infinite system  arising from {two semi-infinite volumes} that are
put in contact by means of a {channel}. For $n\in \N$, we define the $d$-dimensional semi-infinite lattice $\Z^d_{n,+}$ as the set of all points to the right of the hyperplane $x_1=n$
\beq
\Z^d_{n,+}:=\{{\bx}:=(x_1,\ldots,x_d)\in \Z^d: \: x_1\ge n\}.
\eeq
Similarly we define the semi-infinite lattice $\Z^d_{n,-}$ as the set including all points to  the left of the hyperplane $x_1=-n$
\beq
 \Z^d_{-n,-}:=\{{\bx}:=(x_1,\ldots,x_d)\in \Z^d: \: x_1\le -n\}.
\eeq
Finally the channel $C_{n}$ is defined as the centered squared box of side $2n+1$
connecting the two semi-infinite lattices
\beq
C_{n}
:=\{{\bx}\in \Z^d: \; |x_i|\le n, \; \forall i=1,\ldots d\}.
\eeq
The infinite-volume domain is then obtained as the union
\beq
\C:=  \Z^d_{n,+} \cup  C_{n}\cup \Z^d_{-n,-}.
\eeq
Often we shall derive results about the infinite volume by first considering a finite volume
of linear size $N$ and then studying the limit $N\to\infty$.
Thus for all integers $N>n$  we define
\begin{equation}
\label{nota}
\Lambda_{n,N}= \mathcal X_n \cap [-N,N]^{d} \qquad \text{and} \qquad 
S_{n,N} = \Lambda_{n,N+1} \setminus \Lambda_{n,N}.
\end{equation}
We will use the notation $x\sim y$ to denote nearest neighbor sites in $\C$
and $\{x,y\}$ for the un-oriented bond joining them.
%
%
%
%
%
%
%
%
 
\medskip

\subsection{Harmonic functions}

We continue by identifying harmonic functions for our special geometry.
Let  $\{X(t), \; t\ge 0\}$ be the simple symmetric continuos-time random walk on $\C$
which jumps  at rate $1$ to any of its nearest neighbor sites. 
We denote by $\mathbb{P}_x$ the law of this process
started from $X(0) = x$.
The process  is defined by the generator working on 
functions $\psi: \C\to \R$ as
\be
\label{genRW}
G\psi(x) =
\sum_{\substack{y \in \C \\{y\sim x}}}  [\psi(y)-\psi(x)].
\ee
We can interpret \eqref{genRW} as a conservation law
because $\sum_{x\in\C} G\psi(x) =0$ and
then $j_{x\to y}(\psi) = \psi(x) - \psi(y)$
can be interpreted as a ``current''.
When studying the Ginzburg-Landau model we will also
have currents and the main point of
our analysis will be that 
there are stationary measures 
whose average current
is equal to 
$j_{x\to y}(\psi)$
with $\psi$ an harmonic function.

A function $\psi: \C\to \R$ is said to be harmonic if $G\psi(x)=0$ for all $x\in \C$.
 Harmonic functions are stationary for the evolution
defined by \eqref{genRW}.
When studying Fick's law we will be interested 
in currents through a section of the channel.
Thus, for $|\xi |\le n$, we consider
the total flux  $I_\xi$ through a section $\Sigma_\xi$
in the channel
perpendicular to the $x_1$-axis, i.e. 
$\Sigma_\xi = \{x\in C_n \;: \; x \cdot e_1 = \xi\}$
where $e_1$ denotes the unit vector along the $x_1$ axis.
We thus define
\be
I_{\xi}(\psi) = \sum_{x\in\Sigma_\xi} j_{x-e_1 \to x}(\psi)
\ee 
The crucial feature of our geometrical setup is that in dimension
$d \ge 3$ there are non-constant harmonic functions.
We shall say that the random walk $X(\cdot)$ is definitively in a set $A$ if there exists
a finite $T >0$ such that for all $t\ge T$ one has $X(t) \in A$. 

\begin{defin}[The harmonic function $\lambda$]
\label{def111}
We fix $\la^-,\la^+ \in \R$ with $\lambda^- < \lambda^+$
and define a function $\lambda: \C \to \R$ as
\be
\label{harmonic}
{\lambda}_x=\lambda^- \cdot p_{\bx}^{-}+ \lambda^+ \cdot  p_{\bx}^{+},
\ee
with
\be
p_{\bx}^{\pm}= \mathbb P_x(X(\cdot)\in \Z^d_{\pm n,\pm} \: \text{definitively}).
\ee
\end{defin}
\noindent
The following proposition is proved in Appendix A:

\medskip
\begin{prop}
\label{recur}
The function ${\la}$ in Definition \ref{def111} satisfies the following properties.
\begin{enumerate}
\item It is a bounded harmonic function of the process $\{X(t), t \ge 0\}$ with generator $G$.

\smallskip
\item If the spatial dimension $d\ge 3$ then $p^+_x+p^-_x=1$ and 
$\la$ is a non-constant function.

\smallskip
\item The flux $I_{\xi}(\la)$ associated to $\lambda$ 
has the same value for any $|\xi| < n$ and 
$\frac{I_{\xi}(\la)}{n^{d-1}} \le \frac cn$ for some $c>0$.
\end{enumerate}
\end{prop}

\noindent
We will also consider harmonic functions in a finite volume $\Lambda_{n,N} \cup \Delta$
with $\Delta \subseteq S_{n,N}$.
To this aim we introduce the process $\{X^{N,\Delta}(t), t\ge 0\}$  with generator

\begin{equation}
G^{N,\Delta}\psi(x) = 
\left\{
\begin{array}{ll}
\sum_{\substack{y \in \Lambda_{n,N} \\{y\sim x}}}  [\psi(y)-\psi(x)]+ \sum_{\substack{y \in \Delta \\{y\sim x}}}  [\psi(y)-\psi(x)]  \quad
& \text{if }  x\in {\Lambda_{n,N}},\\
&\\
0  
& \text{if } x\in {\Delta}.
\end{array} 
\right.
\end{equation}

\noindent
The process $\{X^{N,\Delta}(t), t\ge 0\}$, taking values in  $\Lambda_{n,N} \cup \Delta$, is a continuos time random walk that jumps at rate 1 to its nearest neighbors
in $\Lambda_{n,N} \cup \Delta$
and is absorbed when it reaches $\Delta$. We call $\tau$ such absorption time.
\begin{defin}[The harmonic function $\lambda^{(N,\Delta,\barr)}$ with boundary condition $\barr$ on $\Delta$]
\label{def222}
We fix  $\Delta\subset S_{n,N}$ (see \eqref{nota}) and $\sigma: \Delta \to \R$
and define a function $\lambda^{(N,\Delta,\barr)}: \Lambda_{n,N} \cup \Delta \to \R$ as
\be
\label{harmonic}
\lambda^{(N,\Delta,\barr)}_x
= \sum_{y\in \Delta} \barr_y \, \mathbb{P}_x(X^{N,\Delta}(\tau) = y)
\ee
Notice that $\lambda^{(N,\Delta,\barr)}_x = \sigma_x$ for $x\in\Delta$.
%
%
%
%
%
\end{defin}
\noindent
While several results of our paper hold true for a general boundary condition $\barr$ on arbitrary set $\Delta \subseteq S_{n,N}$,
two particular cases will be of special interest and are described hereafter.

\medskip

\begin{hyp}[Special settings]
\label{case}
~\\
\noindent
{\bf (a) Fick's law.} In this case
\be
\Delta= \Delta_+ \cup \Delta_- \qquad\text{where}\qquad \Delta_{\pm} = \{y\in S_{n,N} :  y \mp e_1 \in \Lambda_{n,N}\}
\ee
and $\sigma_x = \lambda^{\pm}$ for $x\in \Delta_{\pm}$.
\\
\noindent
{\bf (b) The full setting.}  In this case
\be
\Delta = S_{n,N}
\ee
and $\sigma_x = \lambda_x$ for $x\in \Delta$, where ${\la}$ is the harmonic function in Definition \ref{def111}.
\end{hyp}

%
\begin{rem}
\normalfont
Hypothesis \ref{case}(a) is the natural set-up for the Fick's law,
as discussed in the Introduction.
Under Hypothesis \ref{case}(b) we have that 
$\la_x^{N,\Delta,\sigma} = \lambda_x$ with $x\in\Lambda_{n,N}$ for any integer $N$, see item (3) 
in the proposition below.
This will be used to study the infinite volume dynamics via partial dynamics,
that will be defined in Section \ref{pd-sec}. 
\end{rem}

\noindent
The following proposition is proved in Appendix A:

\begin{prop}[]
\label{recur2}
 
The function $\lambda^{(N,\Delta,\barr)}$ in Definition \ref{def222}
satisfies the following.
\begin{enumerate}
\item
It is an harmonic function with boundary condition $\sigma$ on the set $\Delta$ 
for the process $\{X^{N,\Delta}(t), t\ge 0\}$ with generator $G^{N,\Delta}$. 

\smallskip
\item 
Under Hypothesis  \ref{case}(a) we have $\lim_{N\to\infty} \lambda^{(N,\Delta,\barr)}_x = \lambda_x$.

\smallskip
\item 
Under Hypothesis \ref{case}(b) and for any $N\in \N$ we have $\lambda^{(N,\Delta,\barr)}_x = \la_x$ for $x\in\Lambda_{n,N}$.
\end{enumerate}
\end{prop}

\begin{rem}
\normalfont
Item (2) of Proposition \ref{recur} and item (2) in Proposition \ref{recur2}
show that in the context of Hypothesis
\ref{case}(a) the current $j_{x\to y}(\la^{N,\Delta,\sigma}):= \la_x^{(N,\Delta,\sigma)} - \la_y^{(N,\Delta,\sigma)}$
is not identically zero in the limit $N\to\infty$.
\end{rem}

\medskip

\subsection{Hamiltonian}

As customary in the theory of lattice systems the energy
is given in terms of its potential, thus
the
formal Hamiltonian is
\be
H(\phi) = \sum_{A\in \mathcal{A}} V_A(\phi_A),
\ee
where $\mathcal{A}$ is the set of all finite subsets of the lattice $\C$,
$\phi_A = \{\phi_x\}_{x\in A}$ and $V_A(\phi_A)$ are $C^\infty$ functions.  We may write
$V_A(\phi)$ for  $V_A(\phi_A)$, $\phi_A$ in such a case is the restriction of $\phi$ to $A$.  To study the infinite volume limit we will restrict to the following case:

\noindent
\begin{defin}[Positive interactions.]
\label{Definition 2.1}
By this we mean Hamiltonians which satisfy the following four conditions.

\begin{itemize}

\item $V_A=0$ if the cardinality $|A|$ of $A$ is $\ge 3$, moreover there is $R$ so that $V_{\{x,y\}}=0$  if $|x-y|> R$.

 \item $V_A=V_B$ if $B$ is a translate of $A$.

\item  $V_{\{x,y\}}(\phi)\ge 0$,  $V_{\{x\}}(\phi)  \ge a\phi_x^2-b$, $a>0$.

\item  $V_{\{x,y\}}(\phi)\le c (\phi_x^2+\phi_y^2)$

\end{itemize}

\end{defin}

\medskip

\begin{rem} 
\normalfont
In the first condition we restrict to one and two-body interactions with finite range; in the second one we suppose that the interaction is translational invariant; the third one is special. 
To understand the origin of the third condition it is convenient to consider the typical two-body interaction, that has the form $-c_{x,y}\phi_x\phi_y$.
In the ferromagnetic case $c_{x,y}>0$ so that we can rewrite it as  $\frac12 c_{x,y}(\phi_x-\phi_y)^2 -  \frac12 c_{x,y}(\phi_x^2+\phi_y^2)$.  This means that the one body potential at $x$ has an extra term
$-\frac12 \sum_y c_{x,y}\phi_x^2$, the assumption is then  that, despite this additional term, the one-body potential   is $\ge a\phi_x^2-b$, $a>0$.
Thus the third condition may be seen as a strengthening of the usual super-stability condition for ferromagnetic interactions.  The fourth condition is clearly satisfied in the usual case where the two body interaction has the form $-c_{x,y}\phi_x\phi_y$.
\end{rem}

\begin{rem} 
\normalfont
The stronger super-stability condition is satisfied in the case of quadratic, ferromagnetic two-body interactions and
when the one-body potential  grows as $c \phi_x^4$, $c>0$.  A particular case is the Hamiltonian
\begin{equation}
\label{sinai2}
H(\phi) = \sum_{x} (\phi_x^2 - 1)^2 +  \frac{1}{2} \sum_{x\sim y} (\phi_x - \phi_y)^2.
 \end{equation}
which has  a phase transition
at  small temperatures in $\Z^d, d\ge 2$, as proved by Dinaburg and Sinai \cite{DS}. Indeed the one-body potential has a double-well shape with two
minima at $\pm 1$ and thus forces the charges to be close to $\pm 1$;
the quadratic interaction term forces the charges to be equal.
As a consequence, at low temperatures the Gibbs measure
concentrates on configurations where the charges are mostly
close to $+1$ (or to $-1$).
\end{rem}

\begin{rem} 
\normalfont
Another Hamiltonian  that satisfies the four conditions stated above
 is the quadratic Hamiltonian
\be
\label{quad2}
H(\phi) = \frac12 \sum_{x\in \C} \phi_x^2.
\ee
Here the potentials are only one-body, the interactions are absent.  It is however
interesting because it has almost explicit solutions obtained by using duality.
\end{rem}

\vskip.5cm

We  use the  assumption of positive interactions to study the infinite-volume dynamics.  In  finite volumes we can be much more general.  In the whole sequel $\La$ will denote
a bounded set in $ \C$ and
\be
H_\La(\phi_\La) = \sum_{A\in \mathcal{A}: A \subset  \Lambda } V_A(\phi_A),
\ee
the energy of $\phi_\La$ in $\La$.

\begin{defin}[``General'' interactions]
\label{Definition 2.2}
~\\
\begin{itemize}

\item There are integers $K$ and $R$ so that $V_A=0$ if the cardinality $|A|$ of $A$ is $\ge K$ or if the diameter of $A$ is $> R$.

 \item $V_A=V_B$ if $B$ is a translate of $A$.

\item   There are $a>0$ and $b\ge 0$ so that, for any bounded $\La\in \C$,
\begin{equation}
\label{n2.12}
H_\La(\phi_\La) = H^0_\La(\phi_\La) + H'_\La(\phi_\La),\qquad H^0_\La(\phi_\La)= a \sum_{x\in \La }  \phi_x^2,\qquad
H'_\La(\phi_\La) \ge- b |\La|
\end{equation}
\item There are  $k$ and $c$ so that, for any $A\in \mathcal{A}$ and any $x \in A$,
\begin{equation}
\label{n2.12.0}
|V_A(\phi_A)| +\Big| \frac{\partial}{\partial \phi_x} V_A(\phi_A)\Big|
 + \Big| \frac{\partial^2}{\partial \phi_x^2} V_A(\phi_A)\Big|\le c \sum_{x\in A}\phi_x^{2k}
\end{equation}

\end{itemize}

\end{defin}

\medskip

\noindent
\eqref{n2.12} is the usual super-stability condition which states that the energy is the sum of a stable hamiltonian plus a positive quadratic term.  The assumption on the derivatives in the last condition will be used when studying the dynamics.  To prove the existence of DLR measures in the thermodynamic limit we need more assumptions which are not stated because we will use the above definition only in finite volumes.

When studying dynamics  for general hamiltonians we will first introduce a cutoff,   use it to prove existence and finally show that it can be removed. We  use the following notation: $\La$ and $\Delta$ always denote sets in $\C$, their complement being meant as the complement in $\C$.  Let  $\La$  be a bounded set  $\phi_\La$ and $\phi_{\La^c}$ configurations in $\La$ and its complement, we then set
\begin{equation}
\label{n2.13}
H_\La(\phi_\La| \phi_{\La^c}) = H_\La(\phi_\La) + \sum_{ {\substack{A: A\cap \La \ne \emptyset \\{A\cap \La^c\ne \emptyset}}}} V_A(\phi),\qquad \phi=( \phi_\La, \phi_{\La^c})
\end{equation}

We next introduce the cutoff function $g_R(\xi)$, $\xi \in \mathbb R_+$, $R>1$, by setting
$g_R(\xi)=1$ when $\xi \le R-1$, $g_R(\xi)=0$ when $\xi \ge R$ and $g_R(\xi)$ a decreasing $C^\infty$ function of $\xi$ in $(R-1,R)$ which interpolates between the values 1 and 0.

\medskip
\begin{defin}[``Cutoff Hamiltonians'']
\label{Definition 2.3}
~\\
The general Hamiltonian $H$ with cutoff $R>1$ is:
\be
H_{\La,R}(\phi_\La| \phi_{\La^c}) = H^0_\La(\phi_\La) + g_R(\|\phi_\La\|_2^2) H'_\La(\phi_\La| \phi_{\La^c})
\ee
where $ H^0_\La(\phi_\La)$ is defined in \eqref{n2.12} and
\be
\|\phi_\La\|_2^2 = \sum_{x\in \La} \phi_x^2
\ee

\end{defin}
\medskip
Thus, when $\|\phi_\La\|_2^2>R$, the Hamiltonian $H_{\La,R}(\phi_\La| \phi_{\La^c})$ becomes quadratic with no interaction among charges.

\vskip1cm

\subsection{Dynamics.}
\label{sec-dynamics}
The stochastic Ginzburg Landau model on $\mathcal X_n$ describes the time evolution of
variables $\phi_x(t)$ which represent the amount of ``charge'' at site $x\in \mathcal X_n$
at time $t\ge 0$.
The evolution  is governed by  the infinite system of stochastic differential equations
   \begin{eqnarray}
   \label{1.34.00}
 \phi_x(t)&=&  \phi_x(0)- \int_0^t ds  \sum_{y \in \C, y \sim x}\Big\{ \frac{\partial H}{\partial \phi_x}(\phi(s)) -\frac{\partial H}{\partial \phi_y} (\phi(s))\Big\}\nn\\ &+&
   \beta^{-1/2} \sum_{y\in \C, y \sim x} w_{x,y}(t), \qquad x \in \C
    \end{eqnarray}
where the variables $ w_{x,y}(t)$ are defined in a space $(\Om,P)$ as follows.
An element $\om \in \Om$ is the collection
$\{ B_{x,y}(t), t\ge 0\}$ where $x,y$ run over the pairs $x\sim y$
such that $x<y$ in the lexicographic order.  $P$ is a product measure such that each $\{B_{x,y}(t), t\ge 0\}$ is a standard Brownian motion.   We then set
 \begin{equation}
   \label{1.33.00}
w_{x,y}(t) = B_{\{x,y\}}(t) \; {\rm if} \;x <y,\qquad w_{x,y}(t) = -B_{\{y,x\}}(t) \; {\rm if}\; x >y
    \end{equation}

%

\begin{rem}
\label{hellooo}
\normalfont
We will prove an existence theorem of the dynamics for ``Positive interactions'' (see the previous subsection) and for ``General interactions'' in the finite-volume case that we will describe below.
\end{rem}

\begin{rem}
\normalfont
Restrict the system \eqref{1.34.00} to only two equations, one for $x$ and the other for $y$ with $x \sim y$. By summing the two we see that the total charge $\phi_x(t)+\phi_y(t)$ is conserved thus the process describes exchanges of charges between the two sites.  There is a random white noise term $dw_{x,y}(t)$, 
to which  it is added a drift given by $\{ \frac{\partial H}{\partial \phi_x}(\phi(t)) -\frac{\partial H}{\partial \phi_y} (\phi(t))\}$ that we will call the instantaneous expected
current from $x$ to $y$, which is thus defined as
\be
\label{currdef1}
J_{x \to y} =  \frac{\partial H}{\partial \phi_x} - \frac{\partial H}{\partial \phi_y}.
\ee
\end{rem}


\subsection{Partial dynamics.}
\label{pd-sec}
As mentioned in Remark \ref{hellooo} above we will first study a {\em partial dynamics} where only finitely-many charges (those
contained in a finite volume of linear size $N$) may evolve, while all
the others are frozen at their initial values. The infinite volume
dynamics will then be obtained in the limit $N\to\infty$.

\noindent
The partial dynamics in $\Lambda_{n,N}$  freezes all charges outside $\Lambda_{n,N}$.
We  denote by $\phi$ a configuration in $ \Lambda_{n,N}$ , by $\bar \phi$ a configuration outside $ \Lambda_{n,N}$
and by $(\phi,\bar\phi)$ a configuration in $\C$.
We then write  $\phi^{(N,\Delta,\barr)}(t)=\{\phi^{(N,\Delta,\barr)}_x(t| \phi, \bar \phi, \om, \Delta, \barr), x \in \La_{n,N}\}$  for the solution (when it exists) of
   \begin{eqnarray}
   \label{1.34.00.1}
 \phi^{(N,\Delta,\barr)}_x(t)&=&  \phi_x(0)- \int_0^t ds \bigg( \sum_{y \in \La_{n,N}, y \sim x}\Big\{ \frac{\partial H}{\partial \phi_x}(\phi^{(N,\Delta,\barr)}(s),\bar\phi) -\frac{\partial H}{\partial \phi_y} (\phi^{(N,\Delta,\barr)}(s),\bar\phi)\Big\}\nn\\ &-&
   \sum_{y \in \Delta, y \sim x}\Big\{\frac{\partial H}{\partial \phi_x}(\phi^{(N,\Delta,\barr)}(s),\bar\phi)-\barr_y\Big\}\bigg)+ \beta^{-1/2} \sum_{ y \sim x} w_{x,y},\quad \; x \in \Lambda_{n,N}
    \end{eqnarray}

\medskip

\noindent
We interpret \eqref{1.34.00.1} by saying that at each bond $\{x,y\}$ with $x \in \Lambda_{n,N}$
and $y \in \Delta$ it is attached a reservoir which exchanges charges at a rate dictated
by the chemical potential $\barr_y$.   

\begin{rem}
\normalfont
Under Hypothesis \ref{case}(a)  we are in the setup of the Fick's law and the partial dynamics in \eqref{1.34.00.1} is customary in the analysis of boundary-driven
processes, where the boundary processes simulate external reservoirs attached to the right and left faces of the system and generating currents.  
As we will see in $d\ge 3$ dimensions the currents do not decay as $N\to \infty$ so that the Fick's law is violated in our geometrical setup.
The Hypothesis \ref{case}(b) is used to study the infinite-volume limit.  
The choice of these boundary processes is therefore crucial in our analysis and it is at this point that the harmonic function $\la$ of Definition 
\ref{def111} enters into play.
\end{rem}

\medskip

We close this section by observing that the partial dynamics with a cut-off Hamiltonian is a Markov process, as the following proposition precisely states.
For a general Hamiltonian $H$ we define the differential operator
   \begin{equation}
 \label{GENN2}
L^{n,N, \Delta,\sigma} =\sum_{\substack{x,y\in \Lambda_{n,N}\\{\{x,y\}}}} L_{x,y}+\sum_{ x \in \Lambda_{n,N}}\sum_{\substack{y  \in \Delta \\ {y\sim x}}} \bar{L}_{x,y},
  \end{equation}
which acts on  smooth functions as follows:
  \begin{eqnarray}
\label{gen}
&&L_{x,y}=-\(\frac{\partial H}{\partial \phi_x}-\frac{\partial H}{\partial \phi_y}\)\(\frac{\partial}{\partial \phi_x}-\frac{\partial }{\partial \phi_y}\)+\frac 1 \beta\(\frac{\partial }{\partial \phi_x}-\frac{\partial }{\partial \phi_y}\)^2
   \end{eqnarray}
  \begin{equation}
\label{2.55}
\bar{L}_{x,y}= 
\Big\{-\(\frac{\partial H}{\partial \phi_x}-{\sigma}_y\)\frac{\partial }{\partial \phi_x}+\frac 1 \beta \, \frac{\partial^2}{\partial \phi^2_x}\Big\}.
\end{equation}
\begin{prop}
\label{thmn2.2}
  Let $H$ in \eqref{1.34.00.1} be a cutoff  Hamiltonian 
 (see Definition \ref{Definition 2.3}).  
  Then, for any $\phi$, $\bar\phi , \Delta, \sigma$,
the equation \eqref{1.34.00.1} has solution
$\phi^{(N,\Delta,\sigma)}(t)=\phi^{(N)}(t| \phi, \bar \phi, \om, \Delta,\sigma)$ for $P$-almost all $\om$.  The law of $\{\phi^{(N,\Delta,\sigma)}(t), t\ge 0\}$, defines the transition probability starting from $\phi$
 of a Markov diffusion process whose generator is $L^{n,N, \Delta,\sigma}$ in \eqref{GENN2}.

\end{prop}

\medskip
\noindent
The equations \eqref{1.34.00.1} with the cutoff Hamiltonian have globally Lipschitz coefficients. The proof of Proposition \ref{thmn2.2} then
follows, see for instance the book by Strook and Varadhan \cite{SV}, and Chapter VII,  \S 2 in \cite{RY}.

\vskip.5cm

\section{Main results}
\label{mainresults}

\subsection{Finite volumes}

We fix arbitrarily $n$ and $N>n$, and shorthand $\phi= \{  \phi_x, x \in \La_{n,N}\}$.   We also fix 
$\bar\phi$, $\Delta$,
$\sigma =\{\barr_y\}_{y\in\Delta}$ and  shorthand
 $\lambda^*_x = \lambda^{(N,\Delta,\sigma)}_x$, with  $x\in \Lambda_{n,N}\cup \Delta$, see Definition \ref{def222}.
 Recall that $\la^*_y= \barr_y$ for $y\in\Delta$.
Let $\mu_{n,N,\la^*}(d\phi |\bar \phi)$  be the Boltzmann-Gibbs
measure
\beq\label{muN}
\mu_{n,N,\la^*}(d\phi |\bar \phi)
=
\frac 1 {\mathcal Z_{n,N,\la^*}(\bar{\phi})}  \; \cdot  e^{-\beta [H_\La(\phi|\bar{\phi}) - \sum_{x\in \Lambda_{n,N}}  {\lambda}_x^* \phi_x]}\, d\phi,
\eeq
where $H_\La(\cdot|\bar{\phi})$ is defined in \eqref{n2.13}. The normalizing partition function is
\be
\mathcal Z_{n,N,\la^*}(\bar{\phi}) = \int e^{-\beta [H_\La(\phi|\bar{\phi}) - \sum_{x\in \Lambda_{n,N}}  {\lambda}_x^* \phi_x]}\, d\phi.
\ee
We will prove in Section \ref{pd-section} the following theorem.
\begin{thm}
\label{tobeproved22}
For a general Hamiltonian $H$, let $L^{n,N,\Delta,\sigma}$ be as in \eqref{GENN2} and $f$ a  smooth test function, then 
\be
\label{pippo}
\int (L^{n,N,\Delta,\sigma}f)(\phi) \mu_{n,N,\la^*}(d \phi|\bar\phi) =0.
\ee
Morever, if $H$ in \eqref{1.34.00.1} is a cutoff  hamiltonian (in the sense of Definition \ref{Definition 2.3}),
then  $\mu_{n,N,\la^*}$ is an invariant measure for the partial dynamics.
\end{thm}

We will use Theorem \ref{tobeproved22} to extend the invariance statement to general Hamiltonians.
We denote by $P^{n,N,\la^*}(d\phi d\om|\bar \phi)=\mu_{n,N,\la^*}(d\phi |\bar \phi)\times P(d\om)$ where $P(d\om)$ is the law of the Brownian motions $B_{\{x,y\}}(t)$ used to define the dynamics.
Furthermore we write
$\phi^{(N)}(t)=\{\phi^{(N)}_x(t| \phi, \bar \phi, \om, \Delta, \barr), x \in \La_{n,N}\}$ for the solution (when it exists) of \eqref{1.34.00.1} with initial datum $\phi$.
We will prove in Section \ref{sect5} and Appendix \ref{appB} the following theorem.
\medskip

\begin{thm}
\label{thmn3.1}
With the above notation, for any $\bar\phi$, $\Delta$ and $\barr$ there is a solution $\phi^{(N)}(t)=\phi^{(N)}(t| \phi, \bar \phi, \om, \Delta, \barr),\; t\ge 0$
of \eqref{1.34.00.1} for $P^{n,N,\la^*}$-almost all $(\phi,\om)$.  Moreover
 for any test function $f$
\begin{equation}
\int dP^{n,N,\la^*}  f(\phi^{(N)}(t))  = \int d\mu_{n,N,\la^*} f(\phi)
\end{equation}
Finally, recalling \eqref{currdef1} for notation,
\begin{equation}
\label{n3.4}
\int dP^{n,N,\la^*} J_{x\to y}(\phi^{(N)}(t)) = \la^*_x-\la^*_y
\end{equation}

 \end{thm}


\noindent
{\bf Non-validity of Fick's law}.
With reference to Hypothesis \ref{case}(a), 
and using Propositions \ref{recur} and \ref{recur2}
the above theorem states that, in the limit $N\to\infty$,
the current \eqref{n3.4} is not identically zero, 
against what stated in the
Fick's law.


\medskip
To study the infinite-volume dynamics we will use that,
under Hypothesis \ref{case}(b), $\la^* = \la$ 
and that DLR measures with chemical potential $\la$ 
are invariant under the partial dynamics, a statement
that we specify next.
For this we need more complete notation.  We thus write $\phi$ for a configuration on $\C$, $\phi_{\La}$ and $\phi_{\La^c}$ for its restriction to
$\La_{n,N}$ and $\La_{n,N}^c$.
For $\Delta$ and $\sigma$ as in Hypothesis  \ref{case}(b) we  define an evolution on configurations $\phi$ by setting
\begin{equation}
\label{pluo}
T^{(N)}_t (\phi,\om,\la):= \big (\phi^{(N)}(t|\phi_{\La},\phi_{\La^c},\om,\Delta, \barr),\phi_{\La^c}\big)
\end{equation}
whenever the right hand side is well-defined, the definition being non empty because of Theorem \ref{thmn3.1}.
We will prove in Section \ref{sect5} the following theorem.

\medskip

\begin{thm}
\label{thmn3.1bis}
Let $\mu$ be a DLR measure for the formal hamiltonian $H-\sum_x \la_x \phi_x$ and $\mathcal P
= \mu \times P$.  Then, under Hypothesis \ref{case}(b),
 for any $N>n$ and any test function $f$,
\begin{equation}
\int d\mathcal P  f\circ T^{(N)}_t (\phi,\om,\la)  = \int d\mu  f
\end{equation}

 \end{thm}

Thus the DLR measures are stationary for all partial dynamics.  However the existence of DLR measures  for the general Hamiltonians of Definition \ref{Definition 2.2} is an assumption, more conditions being needed to ensure their existence, for instance those stated in  Definition \ref{Definition 2.1} for  positive interactions.

\vskip.5cm

\subsection{Infinite volume}
We restrict here to positive hamiltonians $H$ in the sense of Definition \ref{Definition 2.1} and for notational simplicity we consider the specific case of the Dinanburg-Sinai hamiltonian defined in \eqref{sinai2}. Let $\la$ be the harmonic function of Definition \ref{def111}, $\mu$ a regular DLR measure with formal hamiltonian
$H - \sum \phi_x  \la_x$.  By regular we mean that it is supported by configurations $\phi$ such that, for all $x$ large enough, $|\phi_x| \le   (\log |x|)^{1/3}$. 

We call $\mathcal P=\mu \times P$ with $P(d\om)$ the law of the Brownians which define the dynamics. We then write $\phi(t) = \phi(t|\phi,\om)$ as the solution (when it exists) of \eqref{1.34.00} with initial datum $\phi$.
 We will prove in Section \ref{sec6} and Appendix \ref{appC} the following result.

\medskip

\begin{thm}
\label{thmn3.2}
With
$\mathcal P$-probability 1 there is a solution $\phi(t)= \phi(t|\phi,\om)$ of \eqref{1.34.00}. For any test function $f$ and any $t>0$
\begin{equation}
\int d\mathcal P  f(\phi(t))  = \int d\mu f(\phi)
\end{equation}
so that $\mu$ is time-invariant.  Finally, recalling \eqref{currdef1} for notation,
\begin{equation}
\label{3.99}
\int dP  J_{x\to y}(\phi(t)) = \la_x-\la_y
\end{equation}

 \end{thm}

 \medskip
 Theorem \ref{thmn3.2} proves the claim, stated in the introduction, that there are stationary measures in infinite volumes  carrying a non zero current.
The theorem will be proved by showing that the solution of the partial dynamics converges, as $N\to \infty$, to $\phi(t)$.

\medskip
\noindent
{\bf Validity of Fick's law.}
Let $\xi\in\{-n,\ldots,n\}$ then the stationary current per unit-area through a section $\Sigma_\xi$ in the channel
is,
by \eqref{3.99},
\be
{\mathcal J }_{\xi}= \frac{1}{n^{d-1}}\sum_{x\in \Sigma_\xi} \Big(\la_{x-e_1} - \la_{x}\Big) = \frac{I_\xi(\la)}{n^{d-1}}
\ee 
By Proposition \ref{recur} it follows that ${\mathcal J }_{\xi}$ does not depend on $\xi$
and it is bounded by $c/n$.
Contrary to what stated after Theorem \ref{thmn3.1}
this shows the validity of the Fick's law if we think of the system as the channel with the semi-spaces $\Z^d_{n,\pm} $ as ``gigantic'' reservoirs.  They  provide a steady current in the channel but, despite that, they do not change  in time: this has evidently to do with the fact that they are infinite, but this is not enough to explain the phenomenon because in $d=2$ the effect is not present.

\vskip.5cm

\subsection{The quadratic hamiltonian}
\label{sec3.3}

The quadratic Hamiltonian in the title is the one defined in \eqref{quad2}.  Being quadratic it may be seen as a cutoff hamiltonian so that the properties stated in Theorem \ref{thmn2.2} apply.  In particular, for any $\phi$, $\bar \phi$, $\bar \la$,
 \eqref{1.34.00.1} has solution
$\phi^{(N)}(t)$  for $P$-almost all $\om$.  Moreover the quadratic hamiltonian fits in the class of positive hamiltonians so that  Theorem \ref{thmn3.2} applies and the infinite-volume dynamics $\phi(t|\phi,\om)$ is well defined with $\mathcal P$-probability 1 and the DLR measure $\mu_\la$ with chemical potential $\la$ (which is a product measure)
is time-invariant.

We have however much more information, in fact, for an Ornestein-Ulhenbeck process 
it is known that Gaussian measures  evolve into Gaussian
measures, so that we only need to determine mean and covariance of the process.
In our case this can be done using duality.
Duality for the quadratic Ginzburg-Landau model follows from the algebraic
approach discussed in \cite{CGR}, see \cite{Groeneveld} for a derivation based on
Lie algebra representation theory. For completeness we shall also provide
a direct proof in Section \ref{due}.

For finite volumes $\Lambda_{n,N}$ duality is stated as follows.
Given $\Delta$ and $\sigma$ the duality function is
\be
\label{hermi}
D^{\Delta,\sigma}(\phi,\eta) := \prod_{x\in \Delta} \sigma_x^{\eta_x}\prod_{x\in \Lambda_{n,N}} h_{\eta_x} (\phi_x)
\ee 
where $h_{n} (\xi)$ with $\xi\in\R$ denotes the Hermite polynomial of degree $n$
and  $\eta\in \N^{\Lambda_{n,N}}$ with $|\eta| = \sum_{x \in \Lambda_{n,N}} \eta_x < \infty$.
Duality relates the Ginzburg-Landau evolution of $\{\phi^N(t), t\ge 0\}$ to the evolution
of the Markov process $\{\eta^{N}(t), t\ge 0\}$ with generator
\be
\mathcal L = \sum_{\substack{x,y\in \Lambda_{n,N}\\{\{x,y\}}}} {\mathcal L}_{x,y} + \sum_{x\in \Lambda_{n,N}} \sum_{y\in \Delta} {\bar{\mathcal{L}}}_{x,y}
\ee
where
\be
(\mathcal L_{x,y}f)(\eta) = \eta_x(f(\eta^{x,y})-f(\eta)) + \eta_y(f(\eta^{y,x})-f(\eta))
\ee
with $\eta^{x,y}$ the configuration obtained from $\eta$ by moving a particle
from site $x$ to site $y$ and
\be
\bar{\mathcal L}_{x,y} f(\eta) = \eta_x(f(\eta^{x,y}) - f(\eta))
\ee
Thus the dual process is made of independent particles with
absorptions at $\Delta$. We denote by ${\mathcal{E}}_\eta$
the expectation with respect to the law of the process $\{\eta^{N}(t), t\ \ge 0\}$
started at $\eta$.
Similarly, we denote by ${{E}}_\phi$
the expectation with respect to the law of the process $\{\phi^{N}(t), t\ge 0\}$
started at $\phi$. We will prove in Section \ref{due} the following result.
\begin{thm}
\label{dual-thm} 
With the above notation we have
\be
\label{dual-formula}
E_{\phi} [D^{\Delta,\sigma}(\phi^N(t),\eta)] = {\mathcal{E}}_\eta [D^{\Delta,\sigma}(\phi,\eta^N(t))]
\ee
\end{thm}

\medskip
\begin{rem}
\normalfont
Using duality, the mean and covariance of the Gaussian 
process $\phi^N(t)$ can be computed starting the dual
process with one and two dual particles.
Furthermore duality also implies convergence
in the limit $t\to\infty$ to the Gibbs measure 
$\mu_{n,N,\la^*}(d\phi)$
given by
\be
\label{n3.1112}
\mu_{n,N,\la^*}(d\phi)=\prod_{x\in\C} \frac 1 {\mathcal Z}  \; \cdot  \exp\Big\{-\tfrac{\beta}{2}  (\phi_x - {\la}^*_x)^2 \Big\}\,  d\phi
\ee
where ${\mathcal Z}$ is a normalizing constant.
Indeed the duality formula \eqref{dual-formula} gives
\be
\lim_{t\to\infty} E_{\phi} [D(\phi(t),\eta)] = \prod_{x\in\Lambda_{n,N}} \Big(\sum_{y\in\Delta} \mathbb{P}_x(X(\infty)=y)\sigma_y\Big)^{\eta_x} = \prod_{x\in\Lambda_{n,N}} (\lambda^*)^{\eta_x}. 
\ee
Expression \eqref{n3.1112} follows by recalling that, for a  Gaussian random variable $Y$ with mean $m$, one has
\be
\mathbb{E}[h_n(Y)] = m^n\;.
\ee
Similarly one can check invariance of \ref{n3.1112}.
Labelling the particles of $\eta$, we describe
$\eta$ as a configuration $X= \{X_i, i=1,\ldots, |\eta|\}$
where the particles evolve independently
and $X_i(t)$ is the position of the $i^{th}$ particle at time $t\ge0$.
We have
\be
\int \mu^{(t)}_{\la}(d\phi) D(\phi,\eta) =  \mathbb{E}_{X} \Big[\prod_{i=1}^{|\eta|} \lambda^*_{X_i(t)} \Big] =  \prod_{i=1}^{|\eta|}  \mathbb{E}_{X_i} \Big[\lambda^*_{X_i(t)}\Big] =  \prod_{i=1}^{|\eta|}   \lambda^*_{X_i}
\ee
where in the last equality it has been used that $\lambda$ is harmonic. Thus
\be
\int \mu^{(t)}_{\la}(d\phi) D(\phi,\eta)
 =  \int \mu_{\la} (d\phi) D(\phi,\eta).
\ee
\end{rem}
\begin{rem}
\normalfont
The duality formula can be used
to characterize the measure at infinite volume 
(by taking the $N\to\infty$ limit in \eqref{dual-formula})
and to show existence of the infinite-volume dynamics
for general initial conditions $\phi$
which may grow polynomially at infinity.
\end{rem}
\begin{rem}
\normalfont
There is a large class of models where  duality
 holds, including both particle systems
(symmetric exclusion, Kipnis-Marchioro-Presutti model, independent particles,
symmetric inclusion) and several interacting diffusions. We refer to the survey in preparation
\cite{CGR}.
Results similar to those of the Ginzburg-Landau model with quadratic Hamiltonian can be obtained
in models where duality holds.
\end{rem}

\vskip.5cm

\section{Proof of Theorem \ref{tobeproved22}}
\label{pd-section}

%
%
%
%

\medskip
\noindent
Equation \eqref{pippo} will be proved via an explicit computation that uses that $\la^*$ is harmonic.
This  generalizes a previous computation by De Masi, Olla, Presutti in \cite{DOP}.
In this section we shorthand $L=L^{n,N,\Delta,\sigma}$.
We have
\begin{eqnarray}
\int (Lf)(\phi) \mu_{n,N,\la}(d\phi)
& = &
\frac 1 {\mathcal Z_{n,N,\la}}  \; \cdot \int (Lf)(\phi)  \, e^{-\beta [H(\phi|\bar{\phi}) - \sum_{x\in\Lambda_{n,N}}  \la^*_x \phi_x]}\, d\phi\nn\\
& = &
\frac {\mathcal Z_{n,0,N}}  {\mathcal Z_{n,N,\la}}  \; \langle Lf, e^{\beta \sum_{x\in\Lambda_{n,N}}  \la^*_x \phi_x} \rangle_{\mu_{n,N,0}} \nn\\
& = & \frac {\mathcal Z_{0,N}}  {\mathcal Z_{\la,N}}  \;  \langle f, L^{\dagger}e^{\beta \sum_{x\in\Lambda_{n,N}}  \la^*_x \phi_x} \rangle_{\mu_{n,N,0}}
\end{eqnarray}
where $\langle f, g\rangle_{\mu_{n,N,0}} :=\int f(\phi)  g(\phi) \, \mu_{n,N,0}(d\phi)$ and $L^{\dagger}$ denotes  the adjoint in $L^2(\mu_{n,N,0})$.
 Hence, to prove \eqref{pippo} it is enough to show that
 \be
 \label{compute}
 (L^{\dagger}g_\la)(\phi)=0
 \qquad
\text{for}
\qquad
 g_\la(\phi):=e^{\beta\sum_{x\in \Lambda_{n,N}}\la^*_{x} \phi_x}.
 \ee
We compute the adjoint
\beq
L^{\dagger} =\sum_{\substack{x,y\in \Lambda_{n,N}\\{\{x,y\}}}} L_{x,y}^\dagger+\sum_{ x \in \Lambda_{n,N}}\sum_{\substack{y  \in \Delta \\ {y\sim x}}} \bar{L}_{x,y}^\dagger.
\eeq
As in \cite{DOP}, we have
\beq
&&L_{x,y}^{\dagger}=L_{x,y}\nn\\
&&\bar{L}_{x,y}^\dagger= e^{\beta \barr_y \sum_{z\in\Lambda_{n,N}} \phi_z}  \cdot \bar{L}_{x,y} \cdot e^{- \beta \barr_y \sum_{z\in\Lambda_{n,N}} \phi_z}
\eeq
Thus we find
\beq
\frac{1}{\beta} (L_{x,y}^\dagger g_\la)(\phi)
& = &
g_\la(\phi)\cdot \[\(\frac{\partial H}{\partial \phi_y}-\frac{\partial H}{\partial \phi_x}\)\(\la^*_x-\la^*_y\)+\(\la^*_x-\la^*_y\)^2\]\nn\\
& = & g_\la(\phi) \cdot (a_x-a_y)(\la^*_x-\la^*_y)\nonumber\\
\eeq
where $x,y\in \Lambda_{n,N}$ and we have defined
\beq
a_x:=\la^*_x-\frac{\partial H}{\partial \phi_x}.
\eeq
Similarly
\beq
\frac{1}{\beta}  (\bar{L}_{x,y}^\dagger g_\la ) (\phi)&= &
g_\la(\phi) \cdot \[\(\barr_{y}-\frac{\partial H}{\partial \phi_x}\)(\la^*_x-\barr_{y})+(\la^*_x-\barr_{y})^2\]\nn\\
&=&
g_\la(\phi)\cdot a_x \cdot (\la^*_{x}-\barr_y)
\eeq
where $x\in \Lambda_{n,N}$ and $y\in \Delta$.
Hence \eqref{compute} is equivalent to
\beq
\sum_{\substack{x,y\in \Lambda_{n,N}\\ {\{x,y\}}}} (a_x-a_y)(\la^*_x-\la^*_y)+\sum_{ x \in \Lambda_{n,N}}\sum_{\substack{y  \in \Delta \\ {y\sim x}}}
a_x(\la^*_x-\barr_{y})  = 0\,.
\eeq
Changing from a sum over bonds to a sum over neighboring sites, we can rewrite this as
\beq
\sum_{x\in \Lambda_{n,N}} a_x \Big[ \sum_{\substack{y\in \Lambda_{n,N}\\ {y \sim x}}} (\la^*_x-\la^*_y)+\sum_{\substack{y  \in \Delta \\ {y\sim x}}} (\la^*_x-\barr_{y}) \Big] = 0\,,
\eeq
which is clearly satisfied as a consequence of Prop. \ref{recur2}.
\qed

\medskip

\vskip.5cm

\section{Proof of Theorems \ref{thmn3.1} and  \ref{thmn3.1bis}}
\label{sect5}

For brevity we call $\La=\La_{n,N}$, and  $\mu(d\phi):= \mu_{n,N,\la^*}(d\phi|\bar\phi)$. We write $ H_{\La,R}(\phi|\bar\phi)$
and $\mu^{(R)}$
when we consider the hamiltonian with cutoff $R$, see Definition \ref{Definition 2.3}.
%
%
 We have already proved that the stochastic differential equations \eqref{1.34.00.1} with the cutoff  hamiltonian  
  have, for any initial datum, global solution with $P$-probability 1, they are denoted here by $\phi^{(R)}(t)$. By what proved in the previous section the Gibbs measure $\mu^{(R)}$ (with the chemical potential $ \la^*$) is invariant.  We  will exploit this to prove a ``time super-stability estimate''.
We write  $\mathcal P^{(R)} = \mu^{(R)} \times P$, and for any configuration $\phi$ in $\La_{n,N}$,
   $$
   \|\phi \|_2^2  = \sum_{x\in \La_{n,N}} \phi_x^2.
   $$
Then we have:
%
%
%
\begin{thm}
\label{thmn1.3}
Given $T>0$ there are $A>0$ and $B$ (independent of $R$)  so that, for all  $S>2$,
   \begin{equation}
   \label{n1.12}
 \mathcal P^{(R)} \bigg[\sup_{t \le T}  \|\phi^{(R)} (t)\|_2^2 \ge S \bigg] \le e^{- AS^2 +B}
    \end{equation}

\end{thm}
Theorem \ref{thmn1.3} will be proved in Appendix \ref{appB}.


\medskip

We will next prove that we can replace $\mu^{(R)}$ by $\mu$ in \eqref{n1.12}.

\medskip

\begin{prop}
Calling
$d\mu^{(R)}(\phi)= G^{(R)}(\phi) d\phi$ and $d\mu(\phi)= G(\phi) d\phi$, we have
  \begin{equation}
   \label{11.4}
 \int d\phi \,| G^{(R)}(\phi)-G(\phi)|   \le 2(p+p')
    \end{equation}
where 
   \begin{equation}
   \label{11.5}
 p = \mu^{(R)} \Big[\|\phi\|_2^2 > R\Big],\qquad p' = \mu\Big[\|\phi\|_2^2 > R\Big]
    \end{equation}
\end{prop}

\medskip
\noindent
{\bf Proof.}
Call $Z^{(R)}$ the partition function then
   \begin{equation*}
Z^{(R)} = \int_{\|\phi\|_2^2 \le R} d\phi \,e^{-\beta H_{\La,R}(\phi|\bar\phi)} + p Z^{(R)}
    \end{equation*}
so that
   \begin{equation*}
G^{(R)}(\phi) = (1-p) \frac {e^{-\beta H_{\La,R}(\phi|\bar\phi)}} {\int_{\|\phi\|_2^2 \le R} d\phi \,e^{-\beta H_{\La,R}(\phi|\bar\phi)}} = (1-p) \frac {e^{-\beta H_{\La}(\phi|\bar\phi)}} {\int_{\|\phi\|_2^2 \le R} d\phi \,e^{-\beta H_{\La}(\phi|\bar\phi)}} ,\quad \|\phi\|_2^2 \le R
    \end{equation*}
The analogous formula holds for $G(\phi') $ so that calling
   $$
   g(\phi):=  \frac {e^{-\beta H_{\La}(\phi|\bar\phi)}} {\int_{\|\phi\|_2^2 \le R} d\phi \,e^{-\beta H_{\La}(\phi|\bar\phi)}} ,\qquad \|\phi\|_2^2 \le R,
   $$
one has
 \begin{eqnarray*}
 \int d\phi \,| G^{(R)}(\phi)-G(\phi)|  & \le & \int_{\|\phi\|_2^2 > R} d\phi \,\big(G^{(R)}(\phi)+G(\phi)\big)+\int_{\|\phi\|_2^2 \le R} d\phi \,| G^{(R)}(\phi)-G(\phi)|   \nn\\ & \le &  p + p'+\int_{\|\phi\|_2^2 \le R} d\phi \,g(\phi)(p+p')
 \le 2(p+p')
    \end{eqnarray*}
hence \eqref{11.4}.  \qed

\medskip

\begin{cor}
There are $A>0$ and $B$ so that, calling $\mathcal A=\{ \sup_{t\le T}  \|\phi^{(R)}(t)\|_2^2 \ge S\}$,
  \begin{equation}
   \label{11.16a}
 \big| (\mu^{(R)}\times P) [\mathcal A] -(\mu\times P) [\mathcal A]\big|\le e^{- A R + B},\quad R> S
    \end{equation}

\end{cor}

\medskip
\noindent
{\bf Proof.} By \eqref{11.4} the left hand side of  \eqref{11.16a} is bounded by
  \begin{equation}
   \label{11.16.0}
   \int dP \int  d\phi |G^{(R)}(\phi)-G(\phi)| \le 2(p+p')
    \end{equation}
    and \eqref{11.16a} follows from \eqref{11.13}.  \qed

    \medskip
    
\noindent
{\em Existence.} 
There are $a'>0$ and $b'$ so that,
    for  $R>S$,
    \begin{equation*}
  (\mu \times P)\Big[\sup_{  t\le T} \|\phi^{(R)} (t)\|_2^2 < S\Big] \ge 1-2e^{-a' S + b'},\quad a'>0
    \end{equation*}
    having used   \eqref{11.16a} and \eqref{n1.12}.
  Therefore, calling $\phi(t)$ the solution of
 \eqref {1.34.00.1}, we have  also
    \begin{equation*}
  (\mu \times P)\Big[\sup_{  t\le T}  \|\phi (t)\|_2^2 < S\Big] \ge 1-2e^{-a' S + b'}
    \end{equation*}
 because $\phi^{(R)} (t)=\phi  (t)$ in the set
 $$\Big\{  \sup_{t\le T}  \|\phi^{(R)} (t)\|_2^2 < S\Big\} ,\quad S< R
 $$
 Thus
   \begin{equation*}
  (\mu \times P)\Big[\sup_{  t\le T}  \|\phi (t)\|_2^2 < \infty\Big] = 1
    \end{equation*}
hence the existence of solutions to \eqref{1.34.00.1}  with probability 1.

\vskip.5cm

\noindent
{\em Time invariance.}  It is enough to prove  that, given any $t>0$,
  \begin{equation}
   \label{11.16.0}
  \int  d\mu(\phi) \int dP  f(\phi(t)) =  \int  d\mu(\phi)  f(\phi)
    \end{equation}
for any test function $f$ such that $\sup_\phi |f(\phi)| \le 1$.

Given any $\eps>0$, let $S$ be such that
  \begin{equation*}
  (\mu \times P)\Big[\sup_{  t\le T}  \|\phi (t)\|_2^2 \ge S \Big] >1-\eps,\quad
  (\mu^{(R)} \times P)\Big[\sup_{  t\le T}  \|\phi^{(R)} (t)\|_2^2 \ge S \Big]>1-\eps
    \end{equation*}
for any $R>S$,  then
 \begin{equation}
   \label{11.16.1}
 \Big| \int  d\mu(\phi) \int dP  f(\phi(t)) -   \int_{\sup_{  t\le T} \|\phi (t)\|_2^2 < S} d\mu(\phi) \times dP  f(\phi(T))\Big| \le \eps
    \end{equation}
For  $R>S$
\begin{equation}
   \label{11.16.2}
    \int_{\sup_{  t\le T} \|\phi (t)\|_2^2 < S} d\mu(\phi) \times dP  f(\phi(T))
    =  \int_{\sup_{  t\le T} \|\phi^{(R)} (t)\|_2^2 < S} d\mu^{(R)}(\phi) \times dP  f(\phi(T))
    \end{equation}
and
 \begin{equation}
   \label{11.16.3}
 \Big| \int  d\mu^{(R)}(\phi) \int dP  f(\phi^{(R)}(T)) -   \int_{\sup_{  t\le T} \|\phi^{(R)} (t)\|_2^2 < S} d\mu(\phi) \times dP  f(\phi^{(R)}(T))\Big| \le \eps
    \end{equation}
Since
 \begin{equation}
   \label{11.16.4}
  \int  d\mu^{(R)}(\phi) \int dP  f(\phi^{(R)}(T)) =   \int  d\mu^{(R)}(\phi)   f(\phi)
    \end{equation}
 we get
   \begin{equation}
   \label{11.16.5}
 \Big| \int  d\mu(\phi) \int dP  f(\phi(T)) - \int  d\mu(\phi)  f(\phi)\Big| \le  2\eps
    \end{equation}
    \qed

\vskip.5cm

\noindent
{\em Average current: proof of \eqref{n3.4}}. 
From time-invariance, 
\be
 \int  d\mu(\phi) \int dP J_{x\to y}(\phi(T)) =  \int  d\mu(\phi)  J_{x\to y}(\phi) 
\ee
Then \eqref{n3.4}  easily follows, using integration
by parts, from the definition of the current \eqref{currdef1} and  the
explicit expression for the stationary measure \eqref{muN}.
%
%
\qed

\vskip.5cm

\noindent
{\em Proof of Theorem \ref{thmn3.1bis}.}
  Fix $N$, we claim that any measure $p$ on $\mathcal X_n$ of the form
$dp(\phi)=d\nu( \phi_{\La_{n,N}^c}) d\mu_{n,N,\la^*}(\phi_{\La_{n,N}}|\phi_{\La_{n,N}^c})$ is invariant under the partial dynamics with $N$.  By choosing
$\nu$ to be the restriction of $\mu$ to configurations on the complement of $\La_{n,N}$ and by using the DLR property, we will then get the invariance statement in the theorem. Let $f(\phi)$ be a smooth test function and let $g_{\phi_{\La_{n,N}^c}} (\phi_{\La_{n,N}}) := f(\phi_{\La_{n,N}},\phi_{\La_{n,N}^c})$.
 By Theorem \ref{thmn3.1} we get
\begin{eqnarray*}
\int p(d\phi)P(d\om) f\circ T^{(N)}_t (\phi,\om,\Delta,\sigma)  &=&\int  \nu(d\phi_{\La_{n,N}^c})\int \mu_{n,N,\la^*}(d\phi_{\La_{n,N}}|\phi_{\La_{n,N}^c}) \int P(d\om)\\&\times& 
g_{\phi_{\La_{n,N}^c}}(\phi^{(N)}(t | \phi_{\La_{n,N}},\phi_{\La_{n,N}^c},\om, \Delta,\sigma))
\\&=&
\int \nu(d\phi_{\La_{n,N}^c})\int \mu_{n,N,\la^*}(d\phi_{\La_{n,N}}|\phi_{\La_{n,N}^c})
g_{\phi_{\La_{n,N}^c}}(\phi_{\La_{n,N}})
\\ &=& \int p(d\phi)f( \phi)
    \end{eqnarray*}
\qed

\vskip.5cm

\section{Proof of Theorem \ref{thmn3.2}}
\label{sec6}

We are in the setup of  Hypothesis \ref{case}(b) so that, by item (3) of Proposition \ref{recur2},  $\la_x^{N,\Delta,\sigma} = \la_x$,
throughout the section $\la$ is the harmonic function in Definition \ref{def111}.
We will use  the following shorthand notation: given $n$ and $N>n$ we denote by $\phi$ a configuration on $\La_{n,N}$, by $\bar \phi$ a configuration in the complement of
$\La_{n,N}$ and by $\mu_{N,\bar \phi,\la}(d\phi)$ the Gibbs measure with hamiltonian $H-\sum_x \phi_x\la_x$  and with boundary condition $\bar\phi$.

The starting point is again a time superstability estimate.
We can not use the one proved in Appendix \ref{appB} because
the parameters in the estimates are volume dependent.
In Appendix \ref{appC} we will first prove 
an equilibrium superstability estimate.
\vskip.5cm
\begin{thm}
\label{nnthm1.1}
There are $a>0$, $N_0>0$ and $b$ so that for all  $N>N_0$ the following holds.
Let  $|\bar\phi_x| \le   (\log |x|)^{1/3}$ for all $x \notin \La_{n,N}$, then for any $x_0
\in \La_{n,N/2}$ and $S>0$
  \begin{equation}
   \label{nn10.2}
 \mu_{N,\bar \phi,\la}\Big[ |\phi_{x_0}| \ge S \Big] \le e^{- aS^4 +b},\quad x_0
\in \La_{n,N/2}
    \end{equation}

\end{thm}

The bound on $\bar\phi$ is motivated by  Corollary \ref{cornn1.2} stated below.

 \medskip
 \begin{defin}
 We set:

 \begin{itemize}

 \item  $M_{\la}$ is the set of all DLR measures $\mu$ with chemical potential $\la$ such that $$\mu\Big[ |\phi_x| \ge S \Big] \le e^{- aS^4 +b}\qquad \text{ for all }\qquad x \in \mathcal X_n.$$
 
     \item  The set $\mathcal G$ of ``good configurations'' is:
      \begin{equation}
   \label{nn1.9}
 \mathcal G = \bigcup_{N> n} {\mathcal G}_{N},\qquad {\mathcal G}_{N}=
 \bigcap_{x \notin \La_{n,N}} \Big\{ |\phi_x| \le (\log |x|)^{1/3}\Big\}
    \end{equation}

 \end{itemize}

 \end {defin}

   \medskip

\begin{cor}
\label{cornn1.2}

With the above notation:

\begin{itemize}


\item  The set $M_\la$ is non-empty because, if $\phi \in \mathcal G$ then, calling $\bar \phi_N$ the restriction of $\phi$ to the complement of $\La_{n,N}$, any weak limit point of $\mu_{N,\bar \phi_N,\la}$ is in $M_{\la}$.

\item If $\mu \in M_{\la}$ then,
 for any $a'<a$, there is $b'$ so that
  \begin{equation}
   \label{nn1.9.0}
 \mu \big[\mathcal G_N \big] \ge 1 -e^{-a'(\log N)^{4/3}+b'}
    \end{equation}
and therefore $\mu[\mathcal G]=1$.

\end{itemize}

\end{cor}

\medskip
\noindent
{\bf Proof.} The first statement follows from  Theorem \ref{nnthm1.1}.
 Since $\mu\Big[ |\phi_x| \ge S \Big] \le e^{- aS^4 +b}$, then
 \begin{equation*}
\mu\Big[|\phi_x| \ge (\log |x|)^{1/3} \Big] \le
e^{- a  (\log |x|)^{4/3} +b}
    \end{equation*}
    which yields \eqref{nn1.9.0}.
       \qed


   \vskip.5cm

We shall next extend the super-stability estimates to the time-dependent case.
Given $n$ and $N$ as above and a boundary configuration $\bar \phi$ we consider the partial dynamics defined in Section \ref{pd-sec}
with hamiltonian \eqref{quad2} and denote by $\{\phi_x(t), x \in \La_{n,N}\}$ the corresponding process. Recall that the charges outside $\La_{n,N}$ are frozen to the initial value $\bar \phi$ and that the dynamics does not depend on the chemical potential $\la$.  We denote by $P_{\mu_{N,\bar \phi, \la}}= \mu_{N,\bar \phi, \la} \times P $ the law of the process $\phi^{(N)}(t|\phi,\om)$ when it starts from $\mu_{N,\bar \phi,\la}$.
In the sequel we fix arbitrarily a positive time $T$ and study the process in the time interval $[0,T]$.  Using that $\mu_{N,\bar \phi,\la}$ is invariant we 
will prove in Appendix \ref{appC} the following:

\medskip

\begin{thm}
\label{thmnn1.3}
Let  $\bar \phi$ be a configuration in $\La_{n,N}^c$ such that $|\bar \phi_x| \le
(\log |x|)^{1/3}$ for all $x\in \La_{n,N}^c$.  Then, given $T>0$ there are $A>0$ and $B$ (independent of  $N$ and $\bar \phi$) so that, for all  $S>2$,
 \begin{equation}
   \label{nn1.3.1}
P_{\mu_{N,\bar \phi,\la}}\Big[\sup_{t \le T}  |\phi_x(t)|  \ge S \Big] \le e^{- AS^4 +B},\quad x \in\La_{n,N/4}
    \end{equation}

\end{thm}

   \vskip.5cm

\noindent
{\bf The infinite-volume limit.}
Here we prove Theorem \ref{thmn3.2}.  We need preliminarily to extend the super-stability estimates from conditional Gibbs measures to DLR measures.

\medskip

\begin{thm}
\label{thmnn1.4}

 For any $N$, any measure $\mu \in M_{\la}$ is invariant under the evolution $T_t^{(N)}(\phi,\om,\la)$,
 see \eqref{pluo}.  Moreover
for all $S>0$,
   \begin{equation}
   \label{nn1.3.2}
 P_{\mu}\Big[\sup_{t \le T}\sup_{x \in \La_{n,N/4}}  |\phi^{(N)}_x(t)| \ge S \Big] \le e^{- AS^4 +B} + e^{- a' (\log N)^{4/3} +b'}
    \end{equation}
where $a'$ and $b'$ are as in \eqref{nn1.9.0}.

\end{thm}

\medskip
\noindent
{\bf Proof. } Invariance has been already proved in Theorem \ref{thmn3.1bis}.
 We condition the probability on the left hand side of \eqref{nn1.3.2} to the configuration $\bar\phi$ outside $\La_{n,N}$.  We can use the bound in \eqref{nn1.3.1} when  $\bar\phi \in \mathcal G_N$ and get in this case the bound with $e^{- AS^4 +B}$. The additional term comes from the contribution of the configurations $\bar \phi$ which are not in $\mathcal G_N$, their probability is bounded using \eqref{nn1.9.0}.
 \qed

   \medskip
   It   follows from Theorem \ref{thmnn1.4} that:

\begin{cor}
Let  $\mu \in M_{\la}$ then
\label{cornn1.5}
 \begin{equation}
   \label{nn1.32}
P_{ \mu}\big[\mathcal G' \big] =1,   \qquad \mathcal G' = \bigcup_{N\ge n}\mathcal G'_N,\qquad
\mathcal G'_N = \Big\{\sup_{t\le T}\sup_{x \in \La_{n,N/4}} |\phi^{(N)}_x(t)| \le (\log N)^{1/3}\Big\} 
    \end{equation}

\end{cor}

 \medskip

\begin{prop}
\label{propnn1.6}
There exist $a''>0$ and $b''$ so that in $\mathcal G'_N$ we have:
  \begin{equation}
   \label{nn1.35.0.1}
\sup_{x \in \La_{n,N_0/8}}\sup_{t\in [0,T]} \big| \phi^{(N)}(t) - \phi^{(2N)}(t)\big|
\le  e^{-a'' N \log N +b''}
    \end{equation}
%
%

\end{prop}

\medskip
\noindent
{\bf Proof.} 
We need to bound
the differences $| \phi^{(N)}_x(t) -  \phi^{(2N)}_x(t)|$ (recall that $2N$ is the first integer after $N$ in the set $\{2^n, n \in \mathbb N\}$)
with $t \le T$ and $N \ge N_0$.
We have:
  \begin{eqnarray}
   \label{nn1.37}
\big| \phi^{(N)}_x(t) -  \phi^{(2N)}_x(t)\big|&\le&  \int_0^t ds  \sum_{y \in \La_N, y \sim x}
 \Big| \Big\{\frac{\partial H}{\partial \phi_x}(\phi^{(N)}(s)) -\frac{\partial H}{\partial \phi_y} (\phi^{(N)}(s))\Big\}\nn\\ &-& \Big\{\frac{\partial H}{\partial \phi_x}(\phi^{(2N)}(s)) -\frac{\partial H}{\partial \phi_y} (\phi^{(2N)}(s))\Big\}
  \Big|
    \end{eqnarray}
The contribution of the two-body potential $V$ to $\frac{\partial H}{\partial \phi_x}$ is uniformly Lipschitz, the one body term is bounded as follows:
  \begin{equation}
   \label{nn1.38}
 \Big| \frac{\partial U}{\partial \phi_x}(\phi_x^{(N)}(s))-\frac{\partial U}{\partial \phi_x}(\phi_x^{(2N)}(s))\Big| \le 12 ((\log 2N)^{1/3})^2 |\phi_x^{(N)}(s)-\phi_x^{(2N)}(s)|
    \end{equation}
because $(\log 2N)^{1/3}$ bounds $|\phi_x^{(i)}(s)|$, $i=N,2N$. 

It then follows:
  \begin{eqnarray}
   \label{nn1.39}
| \phi^{(N)}_x(t) -  \phi^{(2N)}_x(t)|&\le& c \log N\int_0^t ds  \sum_{y \in \La_N, |y - x| \le 3}
| \phi^{(N)}_y(s) -  \phi^{(2N)}_y(s)|
    \end{eqnarray}
We can iterate \eqref{nn1.39} $K$ times with $K$ the largest integer such that $N/8 + 4K \le N/4$.  After $K$  iterations we get \eqref{nn1.35.0.1}.  \qed

 \medskip

 In the theorem below we write $T_t^{(N)}(\phi,\om):=\phi^{(N)}(t|\phi,\om)$.

\begin{thm}
For any $(\phi,\om)\in \mathcal G'$ there is $\phi(t|\phi,\om)$ which satisfies the infinite volume stochastic differential equations \eqref{1.34.00} and,
for any $x$ and any $t\le T$, $\phi^{(N)}_x(t|\phi,\om)$ has a limit when $N\to \infty$, that we denote by $\phi_x(t|\phi,\om)$:
  \begin{equation}
   \label{nn1.35}
\lim_{N\to \infty} \sup_{t \in [0,T]} \big|\phi^{(N)}(t|\phi,\om) -\phi_x(t|\phi,\om)\big| =0
    \end{equation}
Moreover, if $\mu \in M_\la$, then for any test function $f$,
  \begin{equation}
   \label{nn1.36}
 \int \mu(d\phi) \int P(d\om)  f(\phi(t|\phi,\om))  = \int \mu(d\phi)  f(\phi )
    \end{equation}

\end{thm}

\medskip
\noindent
{\bf Proof.} Let $N$ be such that $(\phi,\om) \in \mathcal G'_N$.  Since it satisfies the equations \eqref{1.34.00.1} then $\phi_x^{(N)}(t)$, $t\in [0,T]$, is equi-continuous and bounded and therefore it converges by subsequences to a limit $\phi_x(t)$.  The limit is independent of the subsequence because
$\phi^{(N)}(t)$ is Cauchy by Proposition \ref{propnn1.6}.

By the invariance of $\mu$ for the partial dynamics and   Theorem \ref{thmn3.1bis}, we have, for all $N$,
 \begin{equation}
   \label{nn1.41}
 \mathcal E_{\mu}\Big[ f(\phi^{(N)}(t))\Big]  =   E_{\mu}\Big[ f(\phi (0))
\Big]
    \end{equation}
Then \eqref{nn1.36} follows from \eqref{nn1.35} and the Lebesgue dominated convergence theorem.  \qed

\vskip.5cm

\section{Proof of theorem \ref{dual-thm}}
\label{due}

%

Recalling the notation in Section \ref{sec3.3},
the duality statement \eqref{dual-formula}
is a consequence of the following:
\begin{lem}
\label{lem1}
For $x,y$ nearest neighbors in $\Lambda_{n,N}$ 
we have
$$
L_{x,y} D^{\Delta,\sigma}(\cdot,\eta)(\phi) = {\mathcal L}_{x,y} D^{\Delta,\sigma}(\phi,\cdot)(\eta)
$$
For $x\in \Lambda_{n,N}$ and $y\in \Delta$ nearest neighbors we have
$$
\bar{L}_{x,y} D^{\Delta,\sigma}(\cdot,\eta)(\phi) = {\mathcal{\bar{L}}}_{x,y} D^{\Delta,\sigma}(\phi,\cdot)(\eta)
$$

\end{lem}
\begin{proof}
To alleviate notation we do not write the argument of the polynomials. We have
\begin{eqnarray*}
L_{x,y} D^{\Delta,\sigma}(\cdot,\eta)(\phi)
&=&
\Big[ \prod_{z\in\Lambda_{n,N}, z\neq{x,y}} h_{\eta_z} \Big] \Big[\prod_{z\in \Delta} \sigma_z^{\eta_z} \Big]\\
& &\Big[
h_{\eta_x}'' h_{\eta_{y}}  + h_{\eta_x} h_{\eta_{y}}'' - 2 h_{\eta_x}' h_{\eta_{y}}' \\
& & - \phi_x h_{\eta_x}' h_{\eta_{y}} - \phi_{y} h_{\eta_x} h_{\eta_{y}}' + \phi_x h_{\eta_x} h_{\eta_{y}}' + \phi_y h_{\eta_x}' h_{\eta_{y}}\Big].
\end{eqnarray*}
We regroup terms as follows
\begin{eqnarray*}
L_{x,y} D^{\Delta,\sigma}(\cdot,\eta) (\phi)
&=&
\Big[ \prod_{z\in\Lambda_{n,N}, z\neq{x,y}} h_{\eta_z} \Big] \Big[\prod_{z\in \Delta} \sigma_z^{\eta_z} \Big]\\
& &
\Big[h_{\eta_{x}}' (\phi_{y} h_{\eta_{y}}- h_{\eta_{y}}') + (h_{\eta_x}'' - \phi_x h_{\eta_x}') h_{\eta_{y}}  \\
& &+
(\phi_x h_{\eta_x} - h_{\eta_x}') h_{\eta_{y}}'  + h_{\eta_{x}} (h_{\eta_{y}}'' - \phi_{y} h_{\eta_{y}}')  \Big],
\end{eqnarray*}
and then use the following identities for Hermite polynomials
\be\label{id1}
h_n'(\xi) = n \,h_{n-1}(\xi)
\ee
\be\label{id2}
\xi h_n(\xi) - h_n'(\xi) = h_{n+1}(\xi)
\ee
\be\label{id3}
h_n''(\xi) - \xi h_n'(\xi) = -n h_n(\xi)
\ee
to find
\begin{eqnarray*}
L_{x,y} D^{\Delta,\sigma}(\cdot,\eta)(\phi)
&=&
\Big[ \prod_{z\in\Lambda_{n,N}, z\neq{x,y}} h_{\eta_z} \Big] \Big[\prod_{z\in \Delta} \sigma_z^{\eta_z} \Big]\\
& & \Big[\eta_x (h_{\eta_{x}-1} h_{\eta_{y}+1} - h_{\eta_x}h_{\eta_{y}})
+ \eta_{y}(h_{\eta_x +1} h_{\eta_{y}-1}  - h_{\eta_{x}} h_{\eta_{y}}) \Big] \\
& = &
{\mathcal L}_{x,y} D^{\Delta,\sigma}(\phi,\cdot)(\eta).
\end{eqnarray*}
Similarly, for the boundaries we have
\begin{eqnarray*}
\bar{L}_{x,y} D(\cdot,\eta)(\phi)
&=& 
 \Big[ \prod_{z\in\Lambda_{n,N}, z\neq{x}} h_{\eta_z}\Big] \Big[\prod_{z\in \Delta, z \neq y} \sigma_z^{\eta_z} \Big]\\
&  &
 \Big[ \sigma_y^{\eta_y} h_{\eta_{x}}'' - (\phi_x- \sigma_y)  \sigma_y^{\eta_y} h_{\eta_{x}}' \Big]\\
\end{eqnarray*}
This can be rewritten as
\begin{eqnarray*}
\bar{L}_{x,y} D(\cdot,\eta)(\phi)
&=& 
 \Big[ \prod_{z\in\Lambda_{n,N}, z\neq{x}} h_{\eta_z}\Big] \Big[\prod_{z\in \Delta, z \neq y} \sigma_z^{\eta_z} \Big]\\
& &
 \Big[\sigma_y^{\eta_y+1}h_{\eta_{x}}' + \sigma_y^{\eta_y}( h_{\eta_{x}}'' - \phi_x h_{\eta_{x}}' )\Big]\\
\end{eqnarray*}
By using the identies \eqref{id1} and \eqref{id3} one arrives to
\begin{eqnarray*}
\bar{L}_{x,y} D(\cdot,\eta)(\phi)
&=& 
 \Big[ \prod_{z\in\Lambda_{n,N}, z\neq{x}} h_{\eta_z}\Big] \Big[\prod_{z\in \Delta, z \neq y} \sigma_z^{\eta_z} \Big]\\
 & &
\Big[\eta_x ( \sigma_y^{\eta_y+1}h_{\eta_{x}-1} -  \sigma_y^{\eta_y} h_{\eta_{x}}) \Big]\\
& = &
{\bar{{\mathcal L}}}_{x,y} D(\phi,\cdot)(\eta)
\end{eqnarray*}

\end{proof}

\medskip   
   
\appendix 
\section{}  

\medskip
\noindent
{\bf Proof of Proposition \ref{recur}.} \hspace{0.1cm}  

\medskip
\noindent
{\bf Item (1).} \hspace{0.1cm} 
Introducing the notation $\mathcal{B}^{\pm} = \{X(\cdot)\in \Z^d_{\pm n,\pm} \: \text{definitively}\}$,
we may write
\be
\label{op}
\la_x = \la^+\mathbb{P}_x[\mathcal{B}^{+}] + \la^-\mathbb{P}_x[\mathcal{B}^{-}]
\ee
By letting the walker $X(\cdot)$ do its first jump to one of its neighbors,
and calling $d_x$ the number of neighbors of $x$, we can write
\be
\label{opp}
\mathbb{P}_x[\mathcal{B}^{\pm}] =\frac{1}{d_x} \sum_{\substack{y \in \C \\{y\sim x}}} \mathbb{P}_y[\mathcal{B}^{\pm}] 
\ee  
Inserting \eqref{opp} into \eqref{op} we find
\be
d_x \la_x = \sum_{\substack{y \in \C \\{y\sim x}}} \la_y
\ee
from which it follows that $\la$ is an harmonic function.

\medskip
\noindent
{\bf Item (2).} \hspace{0.1cm} 
The proof follows from classical estimates on the recurrence of
random walks on $\Z^d$. This explains why we need a spatial dimension
larger than two.
In $d\le 2$ the random walk $X(\cdot)$ comes back infinitively many times
to the channel $C_n$ and therefore the only harmonic functions are the
constant ones.  For completeness we give some details.

To show that $p_x^+ + p_x^- = 1$ it is enough to show that if $d\ge 3$ then the random walk
$X(\cdot)$ is definitively in the complement of $C_n$.
Let
\be
K_n = \{x\in\C : x_1=n, \, |x_i| \le n \}
\ee
and define the hitting time of $\{X(t), \; t\ge0\}$ to $K_n$ as
\be
\tau(K_n) = \inf\{t\ge 0 \;:\; X(t) \in K_n\}
\ee
Denoting by $\mathbb{P}^{X}_x$ the law
of the $X(\cdot)$ process started from $x$,
and defining
\be
S_{n,N}^+ =\Z_{n,+}^d \cap S_{n,N-1}\,,
\ee
where $S_{n,N-1}$ has been defined in \eqref{nota},
we claim that $p_x^+ + p_x^- = 1$  is implied by
\be
\label{want}
\lim_{N\to\infty} \min_{x\in S^+_{n,N}}\mathbb{P}^{X}_x[\tau(K_n) = \infty] = 1.
\ee
Indeed, we distinguish the following cases:
\begin{itemize}
\item If $x\in\Lambda_{n,N}$  then with probability 1 the walker
$X(\cdot)$ will hit $S_{n,N}^+\cap S_{n,N}^-$ in a finite time.
\item If $x\in \Z^d_{n,+} \setminus \Lambda_{n,N}$ it can only reach
$K_n$ after passing through $S_{n,N}^+$. The analogous statement
holds for $x\in \Z^d_{n,-} \setminus \Lambda_{n,N}$.
\end{itemize}
To prove \eqref{want}
we call $\{Y(t), \; t \ge 0\}$ the usual continous-time random
walk on $\Z^d$ that jumps with intensity $1$ to any
of its nearest neighboring sites.
Classical estimates prove that, if $d\ge 3$ then  for any compact set $K$,
\be
\label{use}
\lim_{N\to\infty} \min_{x\in S^+_{n,N}} \mathbb{P}^{Y}_x[\tau(K) = \infty] = 1
\ee
where  $\mathbb{P}^{Y}_x$ denotes the law
of the $\{Y(t), \; t \ge 0\}$ process started from $x$.
We may couple $X$ and $Y$ in such a way that
\be
X_i(t) = Y_i(t) \qquad i =2,\ldots,d
\ee
while
\be
X_1(t) =
\left\{
\begin{array}{ll}
Y_1(t) & \text{if } Y_1(t) \ge n,\\
- Y_1(t) +2n -1  & \text{if } Y_1(t) < n.
\end{array} \right.
\ee
If we call $T$ the first time when $X(t)\in K_n$ then $Y(T)\in (K_{n} \cup K_{n-1})$.
Therefore the claim \eqref{want} follows from \eqref{use}
with $K = K_{n} \cup K_{n-1}$.

To prove that $\lambda_x$ is non-constant as $x$ varies in $\C$ we observe that
$p_x^+ \to 1$ and $p_x^-\to 0$ when $x_1\to+\infty$
and the opposite occurs when $x_1\to-\infty$.
This sufficies.

\medskip
\noindent
{\bf Item (3).} \hspace{0.1cm} To show that the flux is the same on each section it is enough to prove that $I_\xi(\la)=I_{\xi+1}(\la)$ for $|\xi| < n$.
 For such a $\xi$ we write, recalling (2.6) and that  $\la$ is harmonic,
       \begin{equation}
   \label{107}
0= \sum_{x\in \Sigma_\xi}G\la(x) = I_\xi(\la)-I_{\xi+1}(\la).
    \end{equation}
 
\smallskip 
\noindent 
We work with  the central section
inside the channel, i.e. $\xi =0$ and prove that for each $x\in\Sigma_0$, we have
$\la_{x-e_1} - \la_x \le c/n$. This in turn follows if we prove that
\begin{equation}
\label{101}
| p_{x^0}^{\pm} - p_{y^0}^{\pm} | \le \frac cn
\end{equation}
where we recall that $p_z^{\pm}$ is defined in (2.9) 
and we take $x^0=(x^0_1,..,x^0_d)$, $y^0=(y^0_1,..,y^0_d)$
with $x^0_1=0$, $y^0_1=-1$ and  $x^0_i=y^0_i$ for $i=2,\ldots,d$.

Call $\{X(t), t\ge 0\}$ and $\{Y(t), t \ge 0\}$ two copies of the random walk process with generator $G$ in \eqref{genRW},
starting respectively from $x^0$ and $y^0$.  We will prove that there exists a coupling $Q$  of these two processes
so that
  \begin{equation}
   \label{102}
Q\Big [ \text{$X(\cdot)$ definitively in $\mathbb Z_{n,\pm}$ and
$Y(\cdot)$ definitively in $\mathbb Z_{n,\mp}$}\Big] \le  \frac {c'}n
    \end{equation}
which clearly implies \eqref{101}.  

To define the coupling $Q$ it is convenient to realize the process $X(t)$ in terms of its coordinates $X_i(t)$.  To each $i \in \{1,\ldots, d\}$ we associate an exponential clock which rings with intensity 2, all clocks are independent.  When a clock rings we take a variable $\varepsilon$ with values $\pm 1$, all the $\varepsilon$-variables are mutually independent.  If the $i$-th clock rings and $\varepsilon$ is the associated variable, then $X_i$ tries to jump: $X_i\to X_i+\varepsilon$, the jump is done if after the jump $X\in \mathcal X_n$, otherwise it is suppressed.  
%

\medskip

{\em Definition of Q:}
The coupling $Q$ is a measure on the sample space $\Omega$ and  we will define $X(t)$ and $Y(t)$  on $\Omega$.
The elements $\omega\in\Omega$ are of the form $\omega=\{t_{n}^{i,x}, \varepsilon_{n}^{i,x}, t_{n}^{1,y}, \varepsilon_{n}^{1,y}\}$ 
where  $\; i \in \{1,\ldots, d\}$  and $n \in \N$.
Under $Q$, the times $t_{n}^{i,x}$  and $t_{n}^{1,y}$ are realizations of 
Poisson processes of intensity $2$ and the increments  $\varepsilon_{n}^{i,x}$  and $\varepsilon_{n}^{1,y}$
are realizations of Bernoulli processes with parameter $1/2$. All these processes
are independent of each other. Thus $Q$ is completely defined.

\medskip

{\em Representation of $X(t)$:} 
We define the processes $\{X_i(t)\}_{t\ge 0}$, with  $i = 1,\ldots, d$,
as the collection of walkers that are initialized from $x^0_i$ 
and at the times $t_{n}^{i,x}$ jumps by $\varepsilon_{n}^{i,x}$
 if the jump is allowed (the walker can not exit $\C$).

\medskip
{\em Representation of $Y(t)$:} 
We first define the auxiliary  processes $\{Y'_i(t)\}_{t\ge 0}$ where  $i = 1,\ldots, d$.
They start from $y^0_i$ and they use the variables $\{t_{n}^{i,x}, \varepsilon_{n}^{i,x}\}$ 
for $i = 2,\ldots, d$ and the variables $\{t_{n}^{1,y}, \varepsilon_{n}^{1,y}\}$ 
for the first coordinate $\{Y'_1(t)\}_{t\ge 0}$.
To define  $\{Y(t), t \ge 0\}$ we introduce the time $\bar{t}$ as 
$$
\bar{t} = \inf \{{t\ge 0 \,:\, Y'_1(t)=X_1(t)}\}
$$
and define $Y(t) = Y'(t)$ for $t\le \bar{t}$.
Then $\{Y(t), t \ge \bar{t}\}$ is constructed by using the variables
$\{t_{n}^{i,x}, \varepsilon_{n}^{i,x}\}$ with $i=1,\ldots,d$
and starting at time $\bar{t}$ from $Y'(\bar{t})$.

\bigskip

Clearly the law of $X(\cdot)$ is $\mathbb{P}_{x^0}$
and the law of $Y(\cdot)$ is $\mathbb{P}_{y^0}$,
thus $Q$ defines the desired coupling.
Having defined the coupling $Q$, we now start the 
analysis of \eqref{102}.
To this aim, it is convenient to define 
two processes $X_1^*(t)$ and $Y_1^*(t)$.
The process $\{X^*_1(t), \;   t \ge 0\}$ 
is initialized from $0$ and, at  times $t_{n}^{1,x}$, jumps by $\varepsilon_{n}^{1,x}$
(with no restrictions).
We define similarly the auxiliary process $\{Z_1(t), \;  t \ge 0\}$:
it starts from $-1$ and it uses the variables $\{t_{n}^{1,y}, \varepsilon_{n}^{1,y}\}$. 
To define $\{Y_1^*(t), t \ge 0\}$ we introduce the time $t^{*}$ as the first time when $Z_1(t)=X^*_1(t)$
and define $Y_1^*(t) = Z_1(t)$ for $t\le t^{*}$.
 $\{Y_1^*(t), t \ge t^{*}\}$ is constructed by using the variables
$\{t_{n}^{1,x}, \varepsilon_{n}^{1,x}\}$ and starting at time $t^{*}$ from $Y'(t^{*})$.
As a consequence, 
$X^*_1(t) = Y^*_1(t)$ for $t\ge t^*$.

%
%
%

\bigskip
We introduce a stopping time $\tau$ as 
 \begin{equation}
   \label{103}
\tau = \inf \Big\{ t \ge 0 \; : \; \max\{|X_1(t)|,|Y_1(t)|\} = n\Big\}
    \end{equation}
and observe that the following three properties hold true:
\begin{itemize}
\item 
\be
\label{uno1}
 X_i(t)  = Y_i(t), \; \text{for} \; i\ge 2,\quad  t\le \tau \; \text{(because this holds at time 0)} 
\ee
\item
\be
\label{due2}
X_1(t) = X_1^*(t) \;\:\text{and} \; \: Y_1(t)=Y_1^*(t) \;\:\text{for} \; \: t\le \tau 
\ee
\item 
\be
\label{treee}
\text{If at  time} \;T<\tau, \:  X(T)=Y(T) \; \text{then} \; X(t) = Y(t) \; \text{for all} \; t>T
\ee
\end{itemize}
Thus, by law of  total   probability, we may write
  \begin{equation}
   \label{105}
 Q[X(T)\ne Y(T)] \le Q[\tau \le T] +  Q[\tau >T;X_1(T)\ne Y_1(T)]
    \end{equation}
where \eqref{uno1} has been used in the second term of the r.h.s.    
By using  \eqref{due2} we have
\be
\{\tau \le T\} = \Big\{\sup_{s\le T} |X_1^*(s)| \ge n\Big\} \cap \Big\{\sup_{s\le T} |Y_1^*(s)| \ge n\Big\}
\ee
By classical estimates for the maximum of a random walk, 
there exist $a>0$ and $b$ so that
      \begin{equation}
   \label{106}
  Q[\tau \le T] \le \frac{e^{-(a n^2/T)+b}}{\sqrt T}
    \end{equation}    
As a consequence of   \eqref{due2} we also have that
\be
Q[\tau >T;X_1(T)\ne Y_1(T)] = Q[\tau >T;X_1^*(T)\ne Y_1^*(T)]
\ee     
Thus 
\be
Q[\tau >T;X_1(T)\ne Y_1(T)] \le Q[X_1^*(T)\ne Y_1^*(T)] \le \frac{c''}{\sqrt T}
\ee        
for some constant $c''$. 
Choosing $T=n^2$   we then get \eqref{102} because,
by \eqref{treee},
  \begin{equation}
   \label{104}
\text{l.h.s. of \eqref{102}} \; \le Q[X(T)\ne Y(T)]
    \end{equation}
 
  \pagebreak
\noindent
{\bf Proof of Proposition \ref{recur2}.} \hspace{0.1cm}  

\medskip
\noindent
{\bf Item (1).} \hspace{0.1cm} We recall that
\be
\label{aaaa3}
\lambda^{(N,\Delta,\barr)}_x
= \sum_{y\in \Delta} \barr_y \, \mathbb{P}_x(X^{N,\Delta}(\tau) = y)
\ee
Similarly to item (1) of Proposition \ref{recur},
by letting the walker $X^{N,\Delta}(\cdot)$ do its first jump 
to one of its neighboring sites,
and calling $d_x$ the number of neighbors of $x\in\Lambda_{n,N}$, 
we can write
\be
\label{opp2}
\mathbb{P}_x(X^{N,\Delta}(\tau) = y) = \frac{1}{d_x} \sum_{\substack{z \in {\Lambda_{n,N}\cup\Delta} \\{z\sim x}}} \mathbb{P}_z(X^{N,\Delta}(\tau) = y) 
\ee  
Inserting \eqref{opp2} into \eqref{aaaa3} we find
\be
d_x \lambda^{(N,\Delta,\barr)}_x = \sum_{\substack{z \in {\Lambda_{n,N}\cup\Delta} \\{z\sim x}}}  \lambda^{(N,\Delta,\barr)}_z
\ee
from which it follows that $G^{N,\Delta}\lambda^{(N,\Delta,\barr)}(x) =0$, i.e. it is an harmonic function.

\medskip
\noindent
{\bf Item (2).} \hspace{0.1cm} 
We are in the setting of Hypothesis \ref{case}(a).
For notational simplicity we assume in this section that the spatial
dimension is fixed to $d=3$.
For any $x$, let $k$ be an integer so that $x\in \Lambda_{n,k}$.
We define
\be
\tau_k = \inf \{t \ge 0 \; :\; X(t) \in S_{n,k}^+ \cup S_{n,k}^-\}
\ee
and we have
\be
\la_x = \sum_{y\in S_{n,k}^+ \cup S_{n,k}^- } \mathbb{P}_x [X(\tau_k) = y] \Big\{ \la^+  \mathbb{P}_y[\mathcal{B}^+] + \la^-  \mathbb{P}_y[\mathcal{B}^-] \Big\}
\ee
We call
\be
\eps_k = \sup_{y\in S_{n,k}^{+} } \mathbb{P}_y[\mathcal{B}^-]
\ee
then
\be
\label{xx1}
\Big|\la_x -  \la^+  \mathbb{P}_x[X(\tau_k) \in S_{n,k}^{+} ] + \la^-  \mathbb{P}_x[X(\tau_k) \in S_{n,k}^{-} ] \Big| \le \eps_k \max(|\la^+|, |\la^-|)
\ee
Similarly, writing $X^N$ for $X^{N,\Delta}$, we take $N> k$ and define
\be
\tau^N = \inf \{t \ge 0 \; :\; |X^N_1(t)| = N\}.
\ee
Then, writing $\la^N$ for $\la^{(N,\Delta,\sigma)}$, we have
\be
\la^N_x = \sum_{y\in S_{n,k}^+ \cup S_{n,k}^- } \mathbb{P}_x [X^N(\tau_k) = y] \Big\{ \la^+  \mathbb{P}_y[X_1^N(\tau^N)=N] + \la^-  \mathbb{P}_y[X_1^N(\tau^N) = -N] \Big\}
\ee
We define
\be
\eps_{k,N} = \sup_{y\in S_{n,k}^{+} } \mathbb{P}_y[X_1^N(\tau^N) = -N]
\ee
and get
\be
\label{xx2}
\Big|\la_x^N -  \la^+  \mathbb{P}_x[X^N(\tau_k) \in S_{n,k}^{+} ] + \la^-  \mathbb{P}_x[X^N(\tau_k) \in S_{n,k}^{-} ] \Big| \le  \eps_{k,N} \max(|\la^+|, |\la^-|).
\ee
Since $X(t) = X^N(t)$ for $t\le \tau_k$ then combining \eqref{xx1} and \eqref{xx2} we find
\be
\label{dis}
|\la_x - \la_x^N | \le  (\eps_k + \eps_{k,N}) \max(|\la^+|, |\la^-|)
\ee
By \eqref{want}  $\eps_k \to 0$ as $k\to \infty$ so that we only need to bound  $\eps_{k,N}$.
Let
\be
D = \{x\in \Lambda_{n,N} \; : \; x_1 = n, |x_i| \le n, \;\text{for} \;i = 2,\ldots, d\}
\ee
and
\be
\tau_D = \inf \{t \ge 0 \; :\; X(t) \in D\}.
\ee
Then
\be
\eps_{k,N} \le \sup_{y \in S_{n,k}^{+}} \mathbb{P}_y [\tau_D < \tau^N] 
\ee
It is convenient to change coordinates $x_1 \to x_1 -n$.
For notational simplicity we rename $N$ instead of $N-n$ and
$k$ instead of $k-n$. 

The aim is to reduce to an estimate on the simple symmetric 
random walk in $\Z^3$,
that we shall call $\{Y(t), t\ge 0\}$.
This can be done by generalizing the argument in the
proof of item (2) of Prop. \ref{recur}.
Thus we introduce a set 
$$
\underline{D} =  \cup_{(m_2, m_3)\in \Z^2} D_{m_2,m_3}
$$
where $D_{0,0} = D$, $D_{m_2\pm 1,m_3} = D_{m_2,m_3} \pm(2N+1)e_2$ which
means that $D_{m_2,m_3}$ is translated by $\pm(2N+1)e_2$.
Analougously $D_{m_2, m_3 \pm 1} = D_{m_2,m_3}  \pm(2N+1)e_3$.
We define
\be
\tau_{\underline{D}} = \inf \{t \ge 0 \; :\; Y(t) \in \underline{D}\},
\ee
\be
\tilde{\tau}^N = \inf \{t \ge 0 \; :\; Y_1(t) \in \{N,-N-1\}\},
\ee
\be
A_{k} = \text{the boundary of} \{x: |x_i| \le k, i=1,\ldots, d\},
\ee
and
\be
A_{k}^+ = A_k \cap \{x: |x_1| \ge 0\}.
\ee
Then
\be
\eps_{k,N} \le \sup_{x\in A_k^+} \mathbb{P}_x[\tau_{\underline D} < \tau_N]
\ee
The proof that $\eps_{k,N}$ vanishes for $k$ and $N$ large follows from:
\begin{prop}
There exist positive constants $c_1,c_2,c_3, a$ such that for all $x\in A_k^+$
\be
\label{hii}
\mathbb{P}_x[\tau_{\underline D} < \tau_N] \le  c_1 \frac{n^{2}}{k} +c_2 \frac{n^2}{\sqrt{N}} + c_3 e^{-a\sqrt{N}}
\ee
\end{prop}
\noindent
{\bf Proof:}
The proof follows from classical estimates for random walks,
in particular the estimate for the Green function 
$G(x,z) \le \frac{c}{|y-z|}$  in 3 dimensions, see \cite{Spi} Chapter 6, \S 26.
Let  
$$
\underline{D}' =  \cup_{\{(m_2, m_3)\in \Z^2 : |m_2| \le \sqrt{N}, |m_3|, \le \sqrt{N}\}} D_{m_2,m_3}
$$
and
$$
\underline{D}'' = \underline{D} \setminus \underline{D}'. 
$$
Then
\be
\mathbb{P}_x[\tau_{\underline D} < \tau_N] \le \mathbb{P}_x[\tau_{{\underline D}'} < \infty] + \mathbb{P}_x[\tau_{{\underline D}''} < \tilde{\tau}^N] 
\ee
We have
\be
\mathbb{P}_x[\tau_{\underline D}' < \infty] \le \sum_{z\in \underline D' } \frac{c}{|x-z|} \le c_1 \frac{n^{2}}{k} +c_2 n^2\frac{\sqrt{N}}{N}
\ee
where the first term bounds the contribution of $D_{0,0}=D$ and the second term
comes from
\be
\sum_{\underline{m}\in {\underline D}' \setminus D_{0,0} } \frac{\tilde{c}n^2}{N |\underline{m}|}
\ee
and 
\be
\mathbb{P}_x[\tau_{{\underline D}''} < \tilde{\tau}^N]  \le \mathbb{P}_x[\tau_{{\underline D}''} \le N^{(1+\frac12)2 - \frac12}] + \mathbb{P}_x[\tilde\tau^N \ge N^{(1+\frac12)2 - \frac12}].
\ee
By classical estimates on the displacement of a random walk we obtain
\be
\mathbb{P}_x[\tau_{{\underline D}''} < \tilde{\tau}^N]  \le c_3 e^{-a\sqrt{N}}
\ee
Hence \eqref{hii} is proved.
\qed
\vskip.3cm
By letting first $N\to\infty$ and then $k\to\infty$ we obtain $\la_x^N \to \la_x$ 
from \eqref{dis}, after recalling that, by \eqref{want},  $\eps_k \to 0$.

\medskip
\noindent
{\bf Item (3).} \hspace{0.1cm} 
We recall that
\be
\la_x = \la^+\mathbb{P}_x[\mathcal{B}^{+}] + \la^-\mathbb{P}_x[\mathcal{B}^{-}]
\ee
so that
\begin{eqnarray}
\la_x 
&=&
\la^+ \sum_{y\in S_{n,N}} \mathbb{P}_y[\mathcal{B}^{+}] \mathbb{P}_x[X(\tau)=y] 
+
\la^- \sum_{y\in S_{n,N}} \mathbb{P}_y[\mathcal{B}^{-}] \mathbb{P}_x[X(\tau)=y]
\nonumber\\
&=&
\sum_{y\in S_{n,N}} \la_y \mathbb{P}_x[X(\tau)=y]
\end{eqnarray}
By using that $X^{N,\Delta}(t) = X(t)$ for $t\le \tau$ we thus find $\la_x = \la_x^{N,\Delta,\sigma}$.

\section{}
\label{appB}
  \medskip
  \noindent
  The following is a weak (volume-dependent) form of Ruelle's superstability estimates.
  \begin{lem}
  There is $c$ (which does depend on $\La$ and $\bar \phi$, but since they are fixed we may regard $c$ as a constant) so that
   \begin{equation}
   \label{11.13}
 \frac{d\mu^{(R)}}{d\phi}(\phi)\equiv G^{(R)}(\phi) \le c  e^{- \beta \frac a2 \| \phi\|_2^2}
    \end{equation}
    \end{lem}

  \medskip

\noindent
{\bf Proof.}  Let $\Delta$ be the finite set of points in $\La^c$  interacting with those in $\La$. Recalling Definition \ref{Definition 2.2} for notation, we write
 $$
H'_{\La}(\phi|\bar\phi) = H'_{\La}(\phi|\bar\phi_\Delta)=
 H'_{\La\cup \Delta}((\phi,\bar\phi_\Delta)) -  H'_{\Delta}(\bar\phi_\Delta)
 $$
 Hence, by \eqref{n2.12},
  $$
H'_{\La}(\phi|\bar\phi) \ge -B | \La\cup \Delta| -|  H'_{\Delta}(\bar\phi_\Delta)|
 $$
The latter term can be regarded as a constant because it only depends on
$\bar\phi$.  Thus $
H'_{\La,R}(\phi|\bar\phi) \ge - g(\| \phi\|_2^2) (|B| | \La\cup \Delta| +|  H'_{\Delta}(\bar\phi_\Delta)|) \ge c$ and therefore
 $$
e^{-\beta \{H_{\La,R}(\phi|\bar\phi)- \sum_{x\in \La} \bar \la_x \phi_x\}} \le
c' e^{-\beta  \sum_{x\in \La} \{ a\phi_x^2 - \bar \la_x \phi_x\}}
\le c'' e^{-\beta \frac a2 \sum_{x\in \La} \phi_x^2}
 $$
We bound from below the partition function by restricting the integral to $|\phi_x| \le 1$ for all $x\in \La$ and we obtain \eqref{11.13}.  \qed

\medskip

To extend the bound to time intervals we will use the following theorem which will be used again  in Appendix \ref{appC}.

\begin{thm}
\label{thmn5.2}
Let $z(t)$, $t \in [0,T]$, $T>0$ be a process with law $P$.  Suppose that for $t \le T$
  \begin{equation}
   \label{n1.15}
z(t) = \int_0^t ds \ga_1(s) + M_t,\qquad M^2_t= M^2_0+ \int_0^t ds \ga_2(s) + N_t
    \end{equation}
with $M_t$ and $N_t$ martingales ($N_0=0$) and that
 \begin{equation}
   \label{n1.15.1}
\sup_{t\le T} \{E[\ga_1(t)^2] + E[\ga_2(t)^2] + E[z(0)^2]\} < \infty
    \end{equation}
    Then
     \begin{equation}
   \label{n1.15.2}
  E\Big[\sup_{t\le T} z^2(t)\Big] \le   2T \int_0^T ds E[\ga_1^2(s)] + 4
  \int_0^T ds E[\ga_2(s)]  + E[z(0)^2]
     \end{equation}

\end{thm}

\medskip

\noindent
{\bf Proof.} Since it is short we give for completeness the proof which can be found in Holley and Strook \cite{hs}  and in De Masi and Presutti \cite{DP}.  We write
 \begin{equation*}
E\Big[\sup_{t \le T}  z^2(t)  \Big] \le  2\Big(  E \Big[\sup_{t \le T} \{ \int_0^t ds \ga_1(s)\}^2  \Big]
 +  E \Big[\sup_{t \le T} M_t^2  \Big]\Big)
    \end{equation*}
By Cauchy-Schwartz
 \begin{eqnarray*}
 &&   E \Big[\sup_{t \le T} \{ \int_0^t ds \ga_1(s)\}^2  \Big]
 \le T E\Big[ \int_0^T ds \ga_1^2(s)   \Big]
= T  \int_0^T ds E \Big[  \ga_1^2(s) \Big]
    \end{eqnarray*}
which is the first term on the right hand side of \eqref{n1.15.2}.

 By Doob's  theorem
  \begin{equation*}
 E\Big[\sup_{t \le T} M_t^2  \Big]  \le 4
 E\Big[M_T^2  \Big]
    \end{equation*}
By  \eqref{n1.15}
 \begin{equation*}
 E\big[M_T^2  \big] =  E\big[
 M^2_0\big]+ T E\big[\ga_2\big]
    \end{equation*}
which completes the proof of the theorem recalling that $M_0=z(0)$.  \qed

\medskip

%
%
\vskip.3cm
\noindent
{\bf Proof of Theorem \ref{thmn1.3}.}  Given $S>2$ we define  a smooth function $f(\xi)$, $\xi\ge 0$, in such a way that $f(\xi)=1$ for $\xi \ge S$
and $f(\xi)=0$ for $\xi \le S-1$.  When $|\xi| \in [S-1,S]$, $f(\xi)$ is a strictly increasing $C^\infty$ function with $0$ derivatives at the endpoints.  As a consequence  $f(\xi)$ is a smooth non decreasing function with derivatives bounded uniformly in $S$, and  if the derivatives $f'(\xi)\ne 0$ or $f''(\xi)\ne 0$
then
$\xi \in [S-1,S]$.

Writing $f_t$ for $f( \|\phi (t)\|_2^2)$  \eqref{n1.12} reads as
   \begin{equation}
   \label{n1.13}
\mathcal P^{(R)}\Big[\sup_{t \le T}  f_t \ge 1 \Big] \le e^{- AS +B}
    \end{equation}
which is implied by
  \begin{equation}
   \label{n1.14}
\mathcal E^{(R)}\Big[\sup_{t \le T}  f^2_t  \Big] \le e^{- AS +B}
    \end{equation}
 We will bound \eqref{n1.14} using Theorem \ref{thmn5.2} with
 $z(t) = f_t$ and
   \begin{equation}
   \label{n1.15'}
\ga_1 = L f,\quad \ga_2 = L f^2 - 2 f Lf
    \end{equation}
 where $L$ is the generator $L^{n,N, \la}$ (with cutoff $R$) of \eqref{GENN2}.
 The role of the measure $P$ in  Theorem \ref{thmn5.2} is now taken by $\mathcal P^{(R)}$.  Since this is time-invariant we get from \eqref{n1.15.2}
     \begin{equation}
   \label{n1.15.2.1}
 \mathcal E^{(R)}\Big[\sup_{t \le T}  f^2_t  \Big] \le    2T^2   E_{\mu^{(R)}}[\ga_1^2] + 4
   T  E_{\mu^{(R)}}[\ga_2]  + E_{\mu^{(R)}}[f^2]
     \end{equation}
  \medskip
     \noindent
     $\bullet$\; Bound of the term with $\ga_1$.

 Recalling \eqref{GENN2}, a contribution to $Lf$ comes from the first order derivatives and it
 is a finite sum of terms of the form
   \begin{equation}
   \label{n1.16}
 \Big\{\frac{ \partial} {\partial \phi_z} [H^0_\La(\phi)+ g_R(\phi) H'_\La(\phi|\bar\phi)]\Big\} \times \frac{ \partial} {\partial \phi_x} f(\phi)
    \end{equation}
while   the second order derivatives give rise to a sum of terms of the form
   \begin{equation}
   \label{n1.17}
  \frac{ \partial^2} {\partial \phi_x^2}  f(\phi) = \frac{ \partial} {\partial \phi_x}\{ 2 \phi_x f'(\phi)\} = 2 f'(\phi) + 4 \phi_x^2 f''(\phi)
    \end{equation}
The key point is that both $|f'|$ and $|f''|$ are bounded by $ \le c \mathbf 1_{\|\phi\|_2^2 \in [S-1,S]}$ and since $g$ and its derivatives are bounded, the expectation of $\ga_1^2$ is bounded by the sum of  finitely-many terms like
    \begin{equation}
   \label{n1.15.2.2}
 E_{\mu^{(R)}}[\pi(\phi)\mathbf 1_{\|\phi\|_2^2 \in [S-1,S] }]  \le
 E_{\mu^{(R)}}[\pi(\phi)^2]^{1/2} \times \mu^{(R)}\big[ \|\phi\|_2^2 \in [S-1,S]\big]^{1/2}
     \end{equation}
  where   $\pi(\phi)$ is a polynomial in $\phi$. 
By \eqref{11.13} this is bounded by $c e^{ - (S-1)(a\beta)/2}$, with $c$ depending on
$\pi(\phi)$, and in conclusion:
    \begin{equation}
   \label{n1.15.2.3}
 E_{\mu^{(R)}}[\ga_1^2]  \le c e^{ - Sa\beta/2}
     \end{equation}
.
 \medskip
     \noindent
     $\bullet$\; Bound of the term with $\ga_2$.

    We have 
         \begin{equation}
   \label{n1.15.2.3.1}
     L_{x,y} f = \beta^{-1} \Big( \frac{ \partial} {\partial \phi_x}-
      \frac{ \partial} {\partial \phi_y}\Big)^2 f,\quad \bar L_{x,y} f=
 \frac{ \partial^2} {\partial \phi_x^2} f
         \end{equation}
then, by \eqref{n1.15'}, we are reduced to the analysis of terms as those considered for $\ga_1$ and we get
    \begin{equation}
   \label{n1.15.2.3}
 E_{\mu^{(R)}}[\ga_2]  \le c e^{ - Sa\beta/2}
     \end{equation}

 \medskip
     \noindent
     $\bullet$\; Bound of the last term in \eqref{n1.15.2.1}.
  We use \eqref{11.13} to write
          \begin{equation}
   \label{n1.15.2.3.4}
  E_{\mu^{(R)}}[f^2] \le c  e^{- \beta \frac a4 (S-1)} \prod_{x\in \La} \int
  e^{- \beta \frac a2 \phi_x^2} \le c' e^{ - S(a\beta)/4}
     \end{equation}
This concludes the proof.
     \qed

\section{}
\label{appC}

We  start by extending the super-stability estimates in \cite{Ruelle-76}, \cite{LP-76} and \cite{BMPP}  to the present case,
namely for the
Hamiltonian \eqref{sinai2} to which it is added the contribution of a chemical potential $\la$ which is a harmonic function.  This
 is a special super-stable hamiltonian where the one body term is
  \begin{equation}
   \label{nn10.1}
 U(\phi_{x}) - \la(x)\phi_x \ge \frac 13 \phi_{x}^4 + \frac 13 \phi_{x}^2 -B
    \end{equation}
(with $B$ a suitable constant) and where the two-body potential, $|\phi_{x}-\phi_y|^2$, is nearest neighbor and evidently non negative.  We will exploit all that to   simplify
the proofs given in the general case.

%

\medskip
    \noindent
{\bf Proof of Theorem \ref{nnthm1.1}.}  Dropping the dependence on $\la$ in the notation we want to bound
 \begin{equation}
   \label{nn10.3}
   \rho(\phi_{x_0}):= \frac 1{Z_{\La_{n,N}}(\bar \phi)}
   \int d\phi_{\La_{n,N}\setminus x_0}
  e^{-\beta [H_{\La_{n,N}}+ W_{\La_{n,N}|\La_{n,N}^c}]}
    \end{equation}
where $H_{\La_{n,N}}$ is the energy in $\La_{n,N}$ (which includes the chemical potential $ \la$)
and $W_{\La_{n,N}|\La_{n,N}^c}$ is the interaction between the charges in $\La_{n,N}$ and those in $\La_{n,N}^c$.

The idea in  \cite{Ruelle-76}, \cite{LP-76} and \cite{BMPP} is to estimate
the integral in \eqref{nn10.3} by introducing a stopping time.  To this end we denote by $\Delta_q$ the cubes of side $2q+1$  centered at $x_0$
taking $q\ge q_0$ where $q_0$ is such that, for $q \ge q_0$,
 \begin{equation}
   \label{nn10.4}
   |\Delta_q| \log q -  |\Delta_{q-1}| \log (q-1) \le 8 d q^{d-1} \log q
    \end{equation}
We choose $N_0$ so that, for $N\ge N_0$, $\Delta_{q_0} \subset \La_{n,N}$.  We often write in the sequel for brevity $\La= \La_{n,N}$, $\La'= \La_{n,N+1}$ and
   $Z_\La(\bar \phi)=Z_{\La_{n,N}}(\bar \phi)$.

We 
  partition the configurations on  $\mathcal X_\La$ into the following atoms:
 \begin{equation}
   \label{nn10.5}
 \mathcal A_0 := \Big\{\phi\in \mathcal X_\La: \sum_{x \in \Delta_{q_0}} \phi_x^2 \le |\Delta_{q_0}|\log q_0  \Big\}
    \end{equation}
and, for $q>q_0$,
\begin{equation}
   \label{nn10.6}
 \mathcal A_q := \Big\{\phi\in \mathcal X_\La: \sum_{x \in \Delta_{q}\cap \La} \phi_x^2 \le |\Delta_{q}|\log q,\;
 \sum_{x \in \Delta_{q'}\cap \La} \phi_x^2 > |\Delta_{q'}|\log q'|, \: q_0 \le q'<q \Big\}
    \end{equation}
Thus $q'$ stops as soon as $\sum_{x \in \Delta_{q}\cap \La} \phi_x^2 \le |\Delta_{q}| \log q$.

We call $\rho_q(\phi_{x_0})$ the integral in \eqref{nn10.3} extended to $\mathcal A_q$ so that
\begin{equation}
   \label{nn10.7}
\rho(\phi_{x_0}) =
\sum_{q\ge q_0}\rho_q(\phi_{x_0})
    \end{equation}
We  split the terms $\rho_q(\phi_{x_0})$ into three classes.

 \vskip.5cm

 $\bullet$\; $\rho_{q_0}(\phi_{x_0})$.  Here we will prove the bound \eqref{nn10.12} below.   We first  drop  the non negative interaction between the charges
in $\Delta_{q_0-1}$ and those in the complement getting
 \begin{equation}
   \label{nn10.8}
\rho_{q_0}(\phi_{x_0}) \le \frac 1{Z_\La(\bar \phi)}\int d\phi_{\La\setminus x_0}
   e^{-\beta [H_{\Delta_{q_0-1}}+ H_{\La \setminus \Delta_{q_0-1}}+ W_{\La |\La^c}]}
    \end{equation}
    because by the assumption on
    $q_0$ there is no interaction between $\Delta_{q_0}$ and the complement of $\La$.
     By \eqref{nn10.1} we have
 \begin{equation}
   \label{nn10.9}
 H_{\Delta_{q_0-1}} \ge   \frac 13 \phi_{x_0}^4 + \frac 13
 \sum_{x \in \Delta_{q_0-1}}\phi_{x}^2
 -B |\Delta_{q_0-1}|
    \end{equation}
 We use the term with $\phi_{x}^2$ to perform the integrals over the variables $\phi_x, x \in \Delta_{q_0-1} \setminus x_0$  so that
 \begin{equation}
   \label{nn10.10}
\rho_{q_0}(\phi_{x_0}) \le \frac 1{Z_\La(\bar \phi)}e^{-  \frac\beta 3 \phi_{x_0}^4 + c|\Delta_{q_0-1}|}\int d\phi_{\La \setminus \Delta_{q_0-1}}
   e^{-\beta [ H_{\La \setminus \Delta_{q_0-1}}+ W_{\La |\La^c}]}
    \end{equation}
To reconstruct a partition function we write
$|\phi_{\Delta_{q_0-1}}| \le 1$ for the set where
$|\phi_x|\le 1$ for all $x \in \Delta_{q_0-1}$.  Then there is $c'$   such that
   $$
e^{c'|\Delta_{q_0}|}
\int_{|\phi_{\Delta_{q_0-1}}| \le 1}
d\phi_{\Delta_{q_0-1}}
  e^{- \beta H_{\Delta_{q_0-1}}}\ge 1
  $$
We claim that $ 2|\Delta_{q_0}|(\log q_0+2d) \ge W_{\La\setminus \Delta_{q_0-1}|\Delta_{q_0-1}} $. Proof: let $x \in \Delta_{q_0-1}$ and $y\in \Delta_{q_0}$, $x\sim y$.  We bound $(\phi_x-\phi_y)^2 \le 2(\phi_x^2
+\phi_y^2)$. By \eqref{nn10.5} the sum over all such $y$ is bounded by $2|\Delta_{q_0}|\log q_0$
 while the sum over all such $x$ is bounded by $2d(|\Delta_{q_0}|- |\Delta_{q_0-1}|)$ hence the claim.
%
%
We then get
\begin{equation}
   \label{nn10.11}
1 \le e^{\beta 2|\Delta_{q_0}|(\log q_0+2d)}\int_{|\phi_{\Delta_{q_0-1}}| \le 1} d\phi_{\Delta_{q_0-1}}
   e^{- \beta[ H_{\Delta_{q_0-1}}+ W_{\La\setminus \Delta_{q_0-1}|\Delta_{q_0-1}}]}
    \end{equation}
By \eqref{nn10.10} and \eqref{nn10.11} we then finally get:
\begin{equation}
   \label{nn10.12}
\rho_{q_0}(\phi_{x_0}) \le c e^{-  \frac\beta 3 \phi_{x_0}^4}
    \end{equation}

 \vskip.5cm

 $\bullet$\; $\rho_{q}(\phi_{x_0})$ with $q$ such that $\Delta_q \subset \La$.   With the same procedure we get the analogue of \eqref{nn10.10}:
  \begin{equation}
   \label{nn10.13}
\rho_{q}(\phi_{x_0}) \le \frac 1{Z_\La(\bar \phi)}e^{-  \frac\beta 3 \phi_{x_0}^4 - \frac\beta 3 \log (q-1) |\Delta_{q-1}| + c|\Delta_{q-1}|}\int d\phi_{\La \setminus \Delta_{q-1}}
   e^{-\beta [ H_{\La \setminus \Delta_{q-1}}+\setminus W_{\La |\La^c}]}
    \end{equation}
 Finally we need  an analogue of \eqref{nn10.11} to reconstruct the partition function.  By \eqref{nn10.4} we have
   \begin{equation}
   \label{nn10.14}
\sum_{y\in \Delta_q \setminus\Delta_{q-1}}\phi_{y}^2 \le  8 d q^{d-1} \log q
    \end{equation}
 and proceeding as before we get
 \begin{equation}
   \label{nn10.15}
\rho_{q}(\phi_{x_0}) \le c' e^{-  \frac\beta 3 \phi_{x_0}^4} e^{ - \frac\beta 3 \log (q-1) |\Delta_{q-1}| + c'|\Delta_{q-1}|+ \beta 8 d q^{d-1} \log q}
    \end{equation}

 \vskip.5cm

 $\bullet$\; $\rho_{q}(\phi_{x_0})$ with $q$ such that $\Delta_q \cap \La'\ne \emptyset$.
 We integrate as before over the charges in $\Delta_{q-1} \cap \La$
 and drop the interaction between $\La \setminus \Delta_{q-1}$ and $\Delta_{q-1} \cap \La$ as well as the interaction between $\Delta_{q-1} \cap \La$ and $\La^c$.
 We then get again
   \begin{equation}
   \label{nn10.16}
\rho_{q}(\phi_{x_0}) \le \frac 1{Z_\La(\bar \phi)}e^{-  \frac\beta 3 \phi_{x_0}^4 - \frac\beta 3 \log (q-1) |\Delta_{q-1}| + c|\Delta_{q-1}|}\int d\phi_{\La \setminus \Delta_{q-1}}
   e^{-\beta [ H_{\La \setminus \Delta_{q-1}}+ W_{\La\setminus \Delta_{q-1}  |\La^c }]}
    \end{equation}
  The reconstruction of the partition function is now more complicated because we need to take into account also the interaction between $\Delta_{q-1} \cap \La$ and $\La^c$.  We call $B^{\rm in}$ the set of points in $\La \setminus \Delta_{q-1}$ which are in a bond with a point in $\Delta_{q-1} \cap \La$. $B^{\rm out}$ instead is the set of points in $\La ^c $ which are in a bond with a point in $\Delta_{q-1} \cap \La$.  Thus we consider the partition function
  \begin{equation}
   \label{nn10.17}
\int_{|\phi_{\Delta_{q-1}\cap \La}| \le 1} d\phi_{\Delta_{q-1}\cap \La}
   e^{- \beta[ H_{\Delta_{q_-1}\cap \La}+ W_{\Delta_{q_-1}\cap \La|B^{\rm in}}
   + W_{\Delta_{q_-1}\cap \La|B^{\rm out}}   ]}
    \end{equation}
 $W_{\Delta_{q_-1}\cap \La|B^{\rm in}}$ is bounded as in  \eqref{nn10.14}.
 \begin{equation*}
  W_{\Delta_{q_-1}\cap \La|B^{\rm out}} =
\sum_{x\in \Delta_{q-1}\cap \La}\sum_{y\in B^{\rm out}}|\phi_x - \bar \phi_y|^2,\qquad |\phi_x| \le 1,\qquad |\bar \phi_y| \le c(\log N)^{1/3}
    \end{equation*}
 because $y\in \La'\setminus \La$.  We have $q \ge N/2$ (because $x_0 \in \La_{n,N/2}$ and the cube of side $2q+1$ and center $x_0$ has non-empty intersection with $\La_{n,N}^c$).  Therefore $\log N \le \log 2q$ so that
 \begin{equation*}
 W_{\Delta_{q_-1}\cap \La|B^{\rm out}}\le   c'' (\log q)^{1/3} q^{d-1}
    \end{equation*}
    In conclusion we get
 \begin{equation*}
\rho_{q}(\phi_{x_0}) \le c' e^{-  \frac\beta 3 \phi_{x_0}^4} e^{ - \frac\beta 3 \log (q-1) |\Delta_{q-1}| + c|\Delta_{q-1}|+ \beta 8 d q^{d-1} \log q + \beta  c'' (\log q)^{1/3} q^{d-1}}
    \end{equation*}

\vskip.5cm
The sum in \eqref{nn10.7} is then bounded as on the right hand side of \eqref{nn10.2} which is therefore proved.  \qed

 \vskip.5cm

\medskip
%
%
%

\medskip

\noindent
   {\bf Proof of Theorem \ref{thmnn1.3}.}
Let $S>2$ and $f(\xi)$, $\xi\ge 0$, be the same as in Appendix \ref{appB}, thus $f(\xi)$ is a smooth non decreasing function with derivatives bounded uniformly in $S$, and  if the derivatives $f'(\xi)\ne 0$ or $f''(\xi)\ne 0$ then
$\xi \in [S-1,S]$.
We fix $x\in \La_{n,N/4}$ and write $f_t$ for $f( |\phi_x(t)|)$. It is then enough to prove that
  \begin{equation}
   \label{nn1.14}
\mathcal E_{\mu_{N,\bar \phi,\la}}\Big[\sup_{t \le T}  f^2_t  \Big] \le e^{- AS^4 +B}
    \end{equation}
$\mathcal E_{\mu_{N,\bar \phi,\la}}$ being the expectation
relative to the process with law $P_{\mu_{N,\bar \phi,\la}}$.

Since $\phi(t)$ solves the stochastic differential equations \eqref{1.34.00.1} we have the following martingale decomposition:
 \begin{equation}
   \label{nn1.15}
f_t = \int_0^t ds \ga_1(s) + M_t,\quad M^2_t= M^2_0+ \int_0^t ds \ga_2(s) + N_t
    \end{equation}
with $M_t$ and $N_t$ martingales ($N_0=0$) and
  \begin{equation}
   \label{nn1.16}
\ga_1(t) = L f_t, \qquad \ga_2(t)= L f^2_t-2f_tLf_t
    \end{equation}
   see for instance Revuz and Yor, \cite{RY}. Since $\mu_{N,\bar \phi,\la}$ is time invariant we get, analogously to \eqref{n1.15.2.1},
       \begin{equation}
   \label{nn1.15.2.1}
\mathcal E_{\mu_{N,\bar \phi,\la}}\Big[\sup_{t \le T}  f^2_t  \Big] \le    2T^2   E_{\mu_{N,\bar \phi,\la}}[\ga_1^2] + 4
   T  E_{\mu_{N,\bar \phi,\la}}[\ga_2]  + E_{\mu_{N,\bar \phi,\la}}[f^2]
     \end{equation}
Recalling that $ x \in\La_{n,N/4}$ and that $f'$ and $f''$ are bounded and
equal to 0 unless $|\phi_x| \in (S-1,S)$, we have
  \begin{equation*}
|L_{x,y} f|\le c S\mathbf 1_{|\phi_x| \in (S-1,S)} \Big( \Big| \frac{\partial H}{\partial \phi_x}\Big| +  \Big|\frac{\partial H}{\partial \phi_y}\Big| + \beta^{-1}
\Big)
    \end{equation*}
Thus, after summing over $|y-x|=1$,
      \begin{equation}
   \label{nn1.16.1.1}
\ga_{1}^2 \le c' S^2\mathbf 1_{|\phi_x| \in (S-1,S)} \Big(\Big|\frac{\partial H}{\partial \phi_x}\Big|^2+ \beta^{-2} + \sum_{|y-x| =1} \Big|\frac{\partial H}{\partial \phi_y}\Big|^2
\Big)
    \end{equation}
Moreover
  \begin{equation*}
\Big| \frac{\partial H}{\partial \phi_z}\Big| \le c \Big( |\phi_z|^3 + \sum_{z':|z'-z| =1}[ |\phi_z|^2+|\phi_{z'}|^2]\Big) \le c \Big( [|\phi_z|^3+2d \phi_z^2] + \sum_{z':|z'-z| =1}|\phi_{z'}|^2 \Big)
    \end{equation*}
    so that
      \begin{equation}
   \label{nn1.16.2}
\Big| \frac{\partial H}{\partial \phi_z}\Big|^2 \le c''   \Big( [|\phi_z|^3+2d \phi_z^2]^2 + \sum_{z':|z'-z| =1}|\phi_{z'}|^4 \Big)
    \end{equation}
Thus
      \begin{equation*}
E_{\mu_{N,\bar \phi,\la}}[\ga_1^2]  \le c'''\max_{z'\ne x, |z'-x| \le 2}E_{\mu_{N,\bar \phi,\la}}\Big[\mathbf 1_{|\phi_x| \in (S-1,S)}\Big( S^8  + S^2  |\phi_{z'}|^4\Big)\Big]
    \end{equation*}
 In conclusion, after using Cauchy-Schwartz and \eqref{nn10.2}
      \begin{equation}
   \label{nn1.16.1.3}
E_{\mu_{N,\bar \phi,\la}}[\ga_1^2]  \le c  \Big( S^8  e^{- aS^4}+  S^2 e^{- (a/2)(S-1)^4 } \Big)
    \end{equation}

Since
      \begin{equation}
   \label{nn1.16.1.4}
 \ga_2 = \beta^{-1} \frac{\partial^2 f}{\partial \phi_x^2}
    \end{equation}
$E_{\mu_{N,\bar \phi,\la}}[\ga_2^2]$ is bounded in a similar way as well as
$E_{\mu_{N,\bar \phi,\la}}[f^2]$, we omit the details.  \qed


%
%
%

\end{document}